\documentclass[sigconf]{acmart}

\usepackage{booktabs} 
\usepackage[utf8]{inputenc}
\usepackage{graphicx,subfigure}
\usepackage{capt-of}
\usepackage[outdir=./]{epstopdf}
\usepackage{paralist}
\usepackage[shortlabels]{enumitem}
\usepackage{amsfonts}
\usepackage{tikz}
\usepackage{bbm}
\usepackage[justification=centering]{caption}

\usepackage{hyperref}
\usepackage{mathtools}
\usepackage{amssymb}
\usepackage{adjustbox}
\usepackage{amsmath}

\newtheorem{thm}{Theorem}[section]
\newtheorem{lem}[thm]{Lemma}

\newcommand\floor[1]{\lfloor#1\rfloor}

\newtheorem{defn}{Definition}[section]

\newtheorem{exmp}{Example}[section]
\usepackage{color,soul}
\usepackage[]{algorithm2e}

\newcommand{\RN}[1]{%
  \textup{\uppercase\expandafter{\romannumeral#1}}%
}
\usepackage{bbm, dsfont}

\setcopyright{rightsretained}

\usepackage{fancyhdr}
\begin{document}
\copyrightyear{2018}
\acmYear{2018}
\setcopyright{acmcopyright}
\acmConference[CCS '18]{2018 ACM SIGSAC Conference on Computer and Communications Security}{October 15--19, 2018}{Toronto, ON, Canada}
\acmBooktitle{2018 ACM SIGSAC Conference on Computer and Communications Security (CCS '18), October 15--19, 2018, Toronto, ON, Canada}
\acmPrice{15.00}
\acmDOI{10.1145/3243734.3243809}
\acmISBN{978-1-4503-5693-0/18/10}
\title{
Preserving Both Privacy and Utility in Network Trace Anonymization}
\author{Meisam Mohammady}
\orcid{1234-5678-9012}
\affiliation{%
  \institution{CIISE, Concordia University, Canada}
}
\email{m_ohamma@encs.concordia.ca.com}
\author{Lingyu Wang}
\affiliation{%
  \institution{CIISE, Concordia University, Canada}
}
\email{wang@ciise.concordia.ca}
\author{Yuan Hong}
\affiliation{%
  \institution{Illinois Institute of Technology}
}
\email{yuan.hong@iit.edu}
\author{Habib Louafi}
\affiliation{%
  \institution{Ericsson Research Security}
}
\email{habib.louafi@ericsson.com}
\author{Makan Pourzandi}
\affiliation{%
  \institution{Ericsson Research Security}
}
\email{makan.pourzandi@ericsson.com}
\author{Mourad Debbabi}
\affiliation{%
  \institution{CIISE, Concordia University, Canada}
}
\email{debbabi@encs.concordia.ca}
\begin{abstract}
As network security monitoring grows more sophisticated, there is an
increasing need for outsourcing such tasks to third-party
analysts. However, organizations are usually reluctant to share their
network traces due to privacy concerns over sensitive information,
e.g., network and system configuration, which may potentially be
exploited for attacks. In cases where data owners are convinced to
share their network traces, the data are typically subjected to
certain anonymization techniques, e.g., CryptoPAn, which replaces real
IP addresses with prefix-preserving pseudonyms. However, most such
techniques either are vulnerable to adversaries with prior knowledge
about some network flows in the traces, or require heavy data
sanitization or perturbation, both of which may result in a
significant loss of data utility. In this paper, we aim to preserve
both privacy and utility through shifting the trade-off from between
privacy and utility to between privacy and computational cost. The key
idea is for the analysts to generate and analyze multiple
anonymized views of the original network traces; those views are
designed to be sufficiently indistinguishable even to adversaries
armed with prior knowledge, which preserves the privacy, whereas one
of the views will yield true analysis results privately retrieved by
the data owner, which preserves the utility. We present the general
approach and instantiate it based on CryptoPAn. We formally analyze
the privacy of our solution and experimentally evaluate it using real
network traces provided by a major ISP. The results show that our
approach can significantly reduce the level of information leakage
(e.g., less than 1\% of the information leaked by CryptoPAn) with
comparable utility.
\end{abstract}



\maketitle
\section{Introduction}
\label{sec:introduction}
As the owners of large-scale network data, today's ISPs and
enterprises usually face a dilemma. As security monitoring and
analytics grow more sophisticated, there is an increasing need for
those organizations to outsource such tasks together with necessary
network data to third-party analysts, e.g., Managed Security Service
Providers (MSSPs)~\cite{outsource}. On the other hand, those
organizations are typically reluctant to share their network trace
data with third parties, and even less willing to publish them, mainly
due to privacy concerns over sensitive information contained in such
data. For example, important network configuration information, such
as potential bottlenecks of the network, may be inferred from network
traces and subsequently exploited by adversaries to increase the
impact of a denial of service attack~\cite{riboni}.

In cases where data owners are convinced to share their network
traces, the traces are typically subjected to some anonymization
techniques. The anonymization of network traces has attracted
significant attention (a more detailed review of related works will be
given in section~\ref{sec:related}). For instance, \textit{CryptoPAn}
replaces real IP addresses inside network flows with prefix preserving
pseudonyms, such that the hierarchical relationships among those
addresses will be preserved to facilitate
analyses~\cite{PP}. Specifically, any two IP addresses sharing a
prefix in the original trace will also do so in the anonymized
trace. However, CryptoPAn is known to be vulnerable to the so-called
\emph{fingerprinting} attack and \emph{injection}
attack~\cite{brekene1,brekene2,a1}. In those attacks, adversaries
either already know some network flows in the original traces (by
observing the network or from other relevant sources, e.g., DNS and
WHOIS databases)~\cite{Burkhart}, or have deliberately injected some
forged flows into such traces.  By recognizing those known flows in
the anonymized traces based on unchanged fields of the flows, namely,
fingerprints (e.g., timestamps and protocols), the adversaries can
extrapolate their knowledge to recognize other flows based on the
shared prefixes~\cite{brekene1}. We now demonstrate such an attack in
details.

\begin{exmp}
\label{exm:1}
In Figure~\ref{tab:motivating-example}, the upper table shows the
original trace, and the lower shows the trace anonymized using
CryptoPAn.  In this example, without loss of generality, we only focus
on source IPs. Inside each table, similar prefixes are highlighted
through similar shading.

\begin{figure}[htb]
	\centering \includegraphics[width=1\linewidth,viewport= 81 160 611 490, clip]{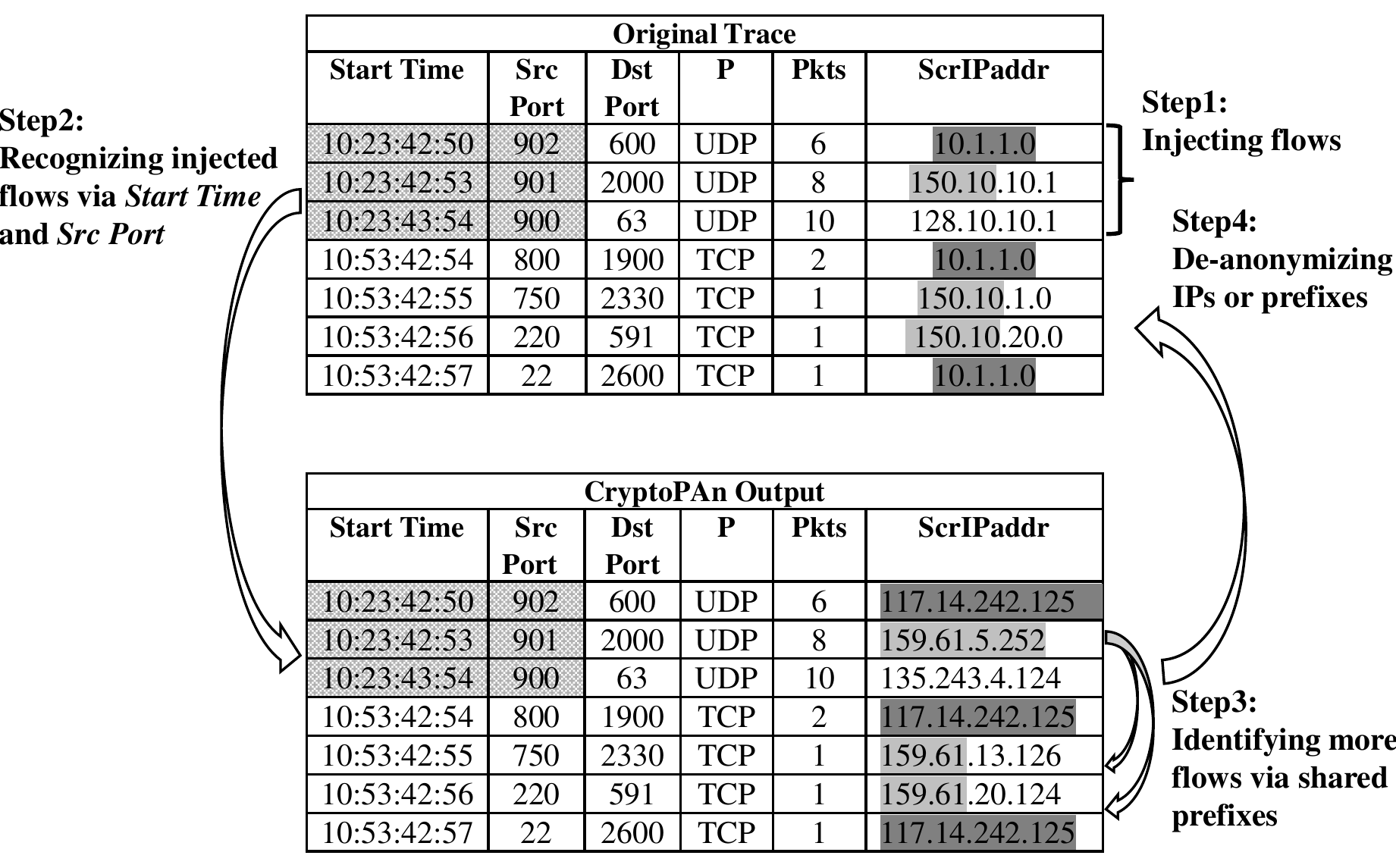}
	\caption{An example of injection attack}
	\label{tab:motivating-example}
\end{figure}

\begin{enumerate}
\item{Step 1:} An adversary has injected three network flows,
  shown as the first three records in the original trace (upper
  table).
\item{Step 2:} The adversary recognizes the three injected flows in
  the anonymized trace (lower table) through unique combinations of
  the unchanged attributes (\emph{Start Time} and \emph{Src
    Port}).
\item{Step 3:} He/she can then extrapolate his/her knowledge from the
  injected flows to real flows as follows, e.g., since prefix $159.61$
  is shared by the second (injected), fifth (real) and sixth (real)
  flows, he/she knows all three must also share the same prefix in the
  original trace. Such identified relationships between flows in the
  two traces will be called \emph{matches} from now on.
\item{Step 4:} Finally, he/she can infer the prefixes or entire IPs of
  those anonymized flows in the original traces, as he/she knows the
  original IPs of his/her injected flows, e.g., the fifth and sixth
  flows must have prefix $150.10$, and the IPs of the fourth and last
  flows must be $10.1.1.0$.
\end{enumerate}
More generally, a powerful adversary who can probe all the subnets
of a network using injection or fingerprinting can potentially
de-anonymize the entire CryptoPAn output via a more sophisticated
\textit{frequency analysis} attack~\cite{brekene1}.
\label{ex:motivating-example}
\end{exmp}
Most subsequent solutions either require heavy data sanitization or
can only support limited types of analysis. In particular, the
$(k,j)$-obfuscation method first groups together $k$ or more flows
with similar fingerprints and then bucketizes (i.e., replacing
original IPs with identical IPs) $j<k$ flows inside each group; all
records whose fingerprints are not sufficiently similar to $k-1$
others will be suppressed~\cite{riboni}. Clearly, both the
bucketization and suppression may lead to significant loss of data
utility. The differentially private analysis method first adds noises
to analysis results and then publishes such aggregated
results~\cite{mcsherry,DP2,DP3}. Although this method may provide privacy
guarantee regardless of adversarial knowledge, the perturbation and
aggregation prevent its application to analyses that demand
accurate or detailed records in the network traces.

In this paper, we aim to preserve both privacy and utility by shifting
the trade-off from between privacy and utility, as seen in most
existing works, to between privacy and computational cost (which has
seen a significant decrease lately, especially with the increasing
popularity of cloud technology). The key idea is for the data owner to
send enough information to the third party analysts such that they can
generate and analyze many different anonymized views of the original
network trace; those anonymized views are designed to be sufficiently
indistinguishable (which will be formally defined in
subsection~\ref{subsec: defn multiview}) even to adversaries armed
with prior knowledge and performing the aforementioned attacks, which
preserves the privacy; at the same time, one of the anonymized views
will yield true analysis results, which will be privately retrieved by
the data owner or other authorized parties, which preserves the
utility. More specifically, our contributions are as follows.

\begin{enumerate}
\item We propose a \emph{multi-view} approach to the prefix-preserving
  anonymization of network traces. To the best of our knowledge, this
  is the first known solution that can achieve similar data utility as
  CryptoPAn does, while being robust against the so-called semantic
  attacks (e.g., fingerprinting and injection). In addition, we
  believe the idea of shifting the trade-off from between privacy and
  utility to between privacy and computational cost may potentially be
  adapted to improve other privacy solutions.
\item In addition to the general multi-view approach, we detail a
  concrete solution based on iteratively applying CryptoPAn to each
  partition inside a network trace such that different partitions are
  anonymized differently in all the views except one (which yields
  valid analysis results that can be privately retrieved by the data
  owner). In addition to privacy and utility, we design the solution
  in such a way that only one \emph{seed} view needs to be sent to the
  analysts, which avoids additional communication cost.
\item We formally analyze the level of privacy guarantee achieved
  using our method, discuss potential attacks and solutions, and
  finally experimentally evaluate our solution using real network
  traces from a major ISP. The experimental results
  confirm that our solution is robust against semantic attacks with a
  reasonable computational cost.
\end{enumerate}

The rest of the paper is organized as follows: Section~\ref{sec:model}
defines our models. Sections~\ref{sec:4} introduces building blocks
for our schemes. Section~\ref{sec:instance} details two concrete
multi-view schemes based on CryptoPAn. Sections~\ref{sec:experiments}
presents the experimental results. Section Appendix~\ref{ffaa} provides
more discussions, and section~\ref{sec:related} reviews the related
work. Finally, section~\ref{sec:conclusion} concludes the paper.

\section{Models} \label{sec:model}

In this section, we describe models for the system and adversaries; we
briefly review CryptoPAn; we provide a high level overview of our
multi-view approach; finally, we define our privacy
property. Essential definitions and notations are summarized in
Table~\ref{tab:math-symbols}.

\begin{table}[hb]
	\caption{The Notation Table.}
	\label{tab:math-symbols}
	\centering
	\begin{adjustbox}{max width=3.3in}
    \LARGE
		\begin{tabular}{c|c|c|c }
			\hline
			Symbol & Definition & Symbol & Definition \\
			\hline
			$\mathcal{L}$ & Original network trace & $\mathcal{L}^*$ & Anonymized trace\\
            \hline
			$A^{\text{IP}}$& IP attributes: source and destination IP & \emph{fp-QI}& Fingerprint quasi identifier\\
			\hline
			$r_i$ & Record number $i$ & $n$& Number of records in $\mathcal{L}$\\
            \hline
			 $\alpha$ & Number of IP prefixes known by the attacker& $\mathcal{S}_{\alpha}$ & The set of addresses known by attacker\\
             \hline
             $\mathcal{S}^*_{0}$ & The set of IP addresses in the seed view & $\mathcal{S}^*_{i}$ & The set of IP addresses in view $i$\\
            \hline
			$PP$ & CryptoPAn function & $RPP$ & Reverse of CryptoPAn\\
            \hline
			$P_i$ & Partition $i$&$m$& Number of partitions in $\mathcal{L}$\\

			\hline
			 $r$& Index of real view&$K_0$, $K_1$ &  Private key and outsourced key\\
			\hline
		\end{tabular}
	\end{adjustbox} \vspace{-5mm}
\end{table}

\subsection{The System and Adversary Model}
\label{subsec:system}
Denote by $\mathcal{L}$ a \textit{network trace} comprised of a set of
\textit{flows} (or records) $r_i$. Each flow includes a confidential
multi-value attribute $A^{\text{IP}}=\{\text{IP}_{src},
\text{IP}_{dst}\}$, and the set of other attributes $A=\{A_i\}$ is
called the \textit{Fingerprint Quasi Identifier (fp-QI)}
\cite{riboni}. Suppose the data owner would like the analyst to
perform an analysis on $\mathcal{L}$ to produce a report $\Gamma$. To
ensure privacy, instead of sending $\mathcal{L}$, an anonymization
function $\mathcal{T}$ is applied to obtain an anonymized version
$\mathcal{L}^*$. Thus, our main objective is to find the anonymization
function $\mathcal{T}$ to preserve both the \emph{privacy}, which
means the analyst cannot obtain $\mathcal{T}$ or $\mathcal{L}$ from
$\mathcal{L^*}$, and \emph{utility}, which means $\mathcal{T}$ must be
prefix-preserving.

In this context, we make following assumptions (similar to those found
in most existing works~\cite{PP,brekene1,brekene2,a1}).
\begin{inparaenum}[i)]
  \item The adversary is a honest-but-curious analyst (in the sense
    that he/she will exactly follow the approach) who can observe
    $\mathcal{L}^*$.
  \item The anonymization function $\mathcal{T}$ is publicly known,
    but the corresponding anonymization key is not known by
    the adversary.
  \item The goal of the adversary is to find all possible matches (as
    demonstrated in Example~\ref{exm:1}, an IP address may be matched
    to its anonymized version either through the fp-QI or shared
    prefixes) between $\mathcal{L}$ and $\mathcal{L}^*$.
  \item Suppose $\mathcal{L}$ consists of $d$ groups each of which
    contains IP addresses with similar prefixes (e.g., those in the
    same subset), and among these the adversary can successfully
    inject or fingerprint $\alpha$ ($\leq d$) groups (e.g., the
    demilitarized zone (DMZ) or other subnets to which the adversary
    has access). Accordingly, we say that the adversary has
    $\mathcal{S}_{\alpha}$ knowledge.
  \item Finally, we assume the communication between the data owner and the analyst is over a secure channel, and we do not consider
    integrity or availability issues (e.g., a malicious adversary may
    potentially alter or delete the analysis report).
\end{inparaenum}

%
\subsection{The CryptoPAn Model}

To facilitate further discussions, we briefly review the CryptoPAn~\cite{PP} model, which gives a baseline for prefix-preserving anonymization.

\begin{defn}\textit{\textbf{Prefix-preserving Anonymization}}~\cite{PP}:
	Given two IP addresses $a = a_1a_2....a_{32}$ and $b =
        b_1b_2....b_{32}$, and a one-to-one function $F(.): \{0,1
        \}^{32} \rightarrow \{0,1 \}^{32}$, we say that
\begin{itemize}
\item[-] $a$ and $b$ share a $k$-bit prefix ($0 \leq k \leq 32$), if and
        only if $a_1a_2....a_k = b_1b_2....b_k$, and $a_{k+1} \neq
        b_{k+1}$.
\item[-] $F$ is prefix-preserving, if, for any $a$ and $b$ that share
  a $k$-bit prefix, $F(a)$ and $F(b)$ also do so.
  \end{itemize}
\end{defn}

Given $a = a_1a_2$ $\cdots a_{32}$ and $F(a)=a'_1a'_2$ $\ldots a'_{32}$, the
prefix-preserving anonymization function $F$ must necessarily satisfy
the canonical form~\cite{PP}, as follows.

\begin{equation}
\label{CryptoPAn}
a'_{i}=a_{i} \oplus f_{i-1}(a_{1}a_{2} \cdots a_{i-1}),~~~i = 1,2,\ldots,32
\end{equation}

where $f_i$ is a cryptographic function which, based on a
$256/128$-bit key $K$, takes as input a bit-string of length $i-1$ and
returns a single bit. Intuitively, the $i^{th}$ bit is anonymized
based on $K$ and the preceding $i-1$ bits to satisfy the
prefix-preserving property. The cryptographic function $f_i$ can be constructed as $L\Big(R
\big(P(a_{1}a_{2} \ldots a_{i-1}),K \big) \Big)$ where $L$ returns the
least significant bit, $R$ can be a block cipher such as
Rijndael~\cite{rindal}, and $P$ is a padding function that expands
$a_1,a_2,\ldots,a_{i-1}$ to match the block size of $R$~\cite{PP}. In
the following, $PP$ will stand for this CryptoPAn function and its
output will be denoted by $a'=PP(a,K)$.

The advantage of CryptoPAn is that it is deterministic and allows
consistent prefix-preserving anonymization under the same
$K$. However, as mentioned earlier, CryptoPAn is vulnerable to
semantic attacks, which will be addressed in next section.
\subsection{The Multi-View Approach}


\label{subsec:Overview}
We propose a novel \emph{multi-view} approach to the prefix-preserving
anonymization of network traces. The objective is to preserve both the
privacy and the data utility, while being robust against semantic
attacks. The key idea is to hide a prefix-preserving anonymized view,
namely, the \emph{real view}, among $N-1$ other \emph{fake views},
such that an adversary cannot distinguish between those $N$ views,
either using his/her prior knowledge or through semantic attacks. Our
approach is depicted in Figure~\ref{fig:2} and detailed below.
\begin{figure*}[tb] \centering
	\includegraphics[width=0.75\linewidth, viewport=35 45 775 460,clip]{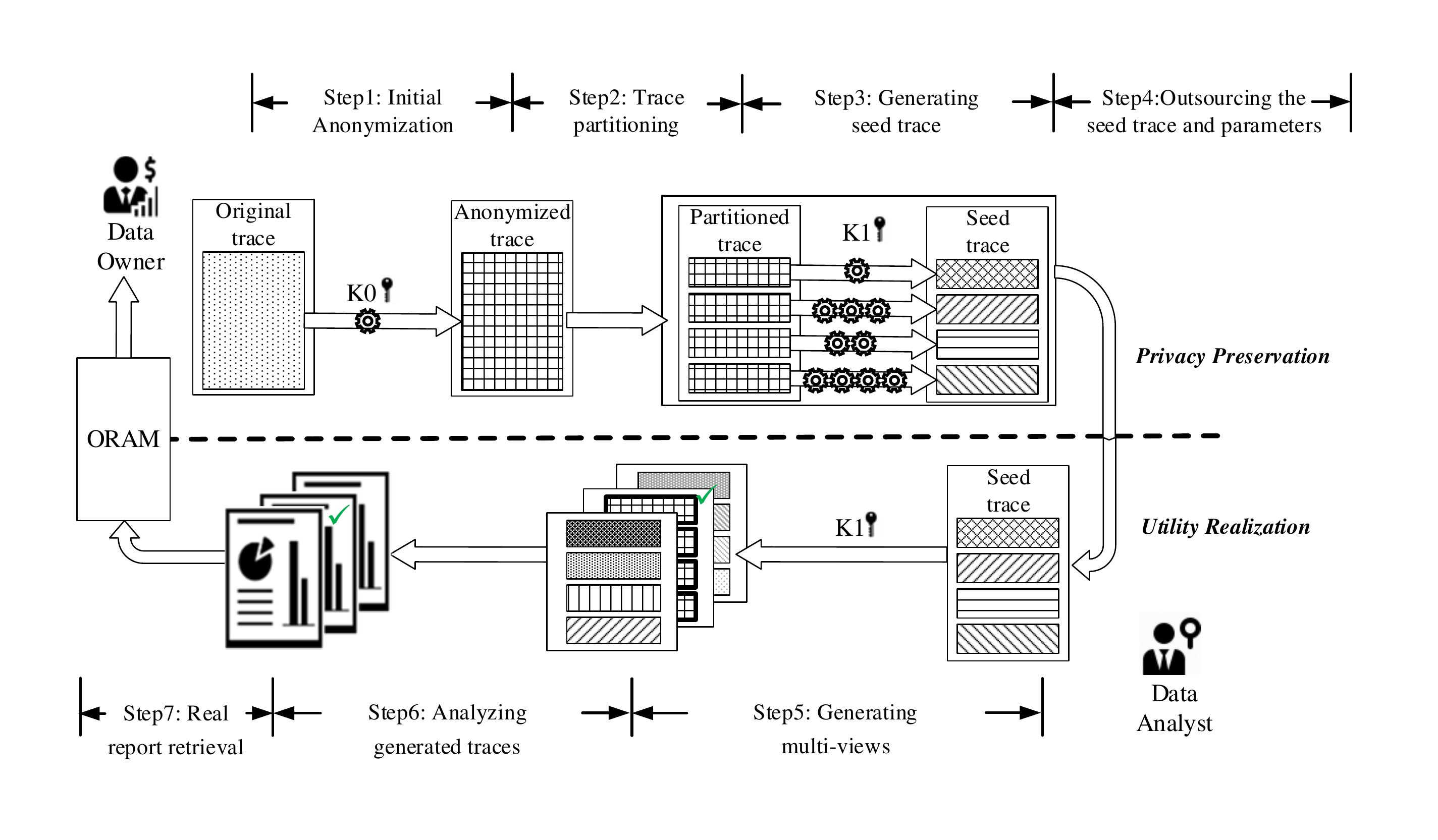}
\caption{An overview of the multi-view approach}
\label{fig:2}
\end{figure*}
\subsubsection{Privacy Preservation at the Data Owner Side}
\begin{description}
	\item[Step 1:] The data owner generates two CryptoPAn keys
          K$_0$ and K$_1$, and then obtains an anonymized trace using
          the anonymization function $PP$ (which will be
          represented by the gear icon inside this figure) and
          K$_0$. This initial anonymization step is designed to
          prevent the analyst from simulating the process as K$_0$
          will never be given out. Note that this anonymized trace is still vulnerable to semantic attack and must undergo the remaining steps. Besides, generating this anonymized trace will actually be slightly more complicated due to \emph{migration} as discussed later in Section~\ref{sec:migration sec}.
	\item[Step 2:] The anonymized trace is then partitioned (the
          partitioning algorithms will be detailed in
          Sections~\ref{sec:pbpp} and \ref{sec:instance}).
        \item[Step 3:] Each partition is anonymized using $PP$ and key K$_1$, but the
          anonymization will be repeated, for a different number of
          times, on different partitions. For example, as the figure
          shows, the first partition is anonymized only once, whereas
          the second for three times, etc. The result of this step is
          called the \textit{seed trace}. The idea is that, as
          illustrated by the different graphic patterns inside the
          seed trace, different partitions have been anonymized
          differently, and hence the seed trace in its entirety is no
          longer prefix-preserving, even though each partition is
          still prefix-preserving (note that this is only a simplified
          demonstration of the \emph{seed trace generator scheme}
          which will be detailed in Section~\ref{sec:instance}).
	\item[Step 4:] The seed trace together with some supplementary
          parameters, including K$_1$, are outsourced to the analyst.
          \end{description}
          \subsubsection{Utility Realization at the Data Analyst Side}
          \begin{description}
	\item[Step 5:] The analyst generates totally $N$ \emph{views} based on the received seed view and supplementary
          parameters. Our design will ensure one of those generated
          views, namely, the \emph{real view}, will have all its
          partitions anonymized in the same way, and thus be
          prefix-preserving (detailed in Section~\ref{sec:instance}),
          though the analyst (adversary) cannot tell
          which exactly is the real view.
	\item[Step 6:] The analyst performs the analysis on all the
          $N$ views and generates corresponding reports.
	\item[Step 7:] The data owner retrieves the analysis report
          corresponding to the real view following an oblivious random
          access memory (ORAM) protocol~\cite{oram}, such that the
          analyst cannot learn which view has been retrieved.
\end{description}
Next, we define the privacy property for the multi-view solution.

\subsection{Privacy Property against Adversaries}
\label{subsec: defn multiview}

Under our multi-view approach, an analyst (adversary) will receive $N$
different traces with identical fp-QI attribute values and different
$A^{IP}$ attribute values. Therefore, his/her goal now is to identify
the real view among all the views, e.g., he/she may attempt to observe
his/her injected or fingerprinted flows, or he/she can launch the
aforementioned semantic attacks on those views, hoping that the real
view might respond differently to those attacks.  Therefore, the main
objective in designing an effective multi-view solution is to satisfy
the \emph{indistinguishability} property which means the real view
must be sufficiently indistinguishable from the fake views under
semantic attacks. Motivated by the concept of \emph{Differential
  Privacy}~\cite{dworks}, we propose the $\epsilon$-indisinguishablity
property as follows.

\begin{defn}\textit{\textbf{$\epsilon$-Indisinguishable Views:}}
\label{dfn:d.2}
A multi-view solution is said to satisfy
\emph{$\epsilon$-Indistinguishability} against an
$\mathcal{S}_{\alpha}$ adversary if and only if (both probabilities
below are from the adversary's point of view)
\begin{align}
\label{eq:mms}
 \exists \ \epsilon \geq 0, \ \textrm{s.t.} \ \forall i \in \{1,2,
 \cdots, N\} \Rightarrow \nonumber \\ & \hspace{-5.1cm} e^{-\epsilon}
 \leq \frac{Pr(\text{view $i$ may be the real view
     })}{Pr(\text{view $r$ may be the real view})} \leq
 e^{\epsilon}
\end{align}
\end{defn}


In Defintion~\ref{dfn:d.2}, a smaller $\epsilon$ value is more
desirable as it means the views are more indistinguishable from the real view to the adversary. For example, an extreme case of $\epsilon=0$
would mean all the views are equally likely to be the real view to the
adversary (from now on, we call these views the \emph{real view candidates}). In practice, the value of $\epsilon$ would depend on the
specific design of a multi-view solution and also on the adversary's
prior knowledge, as will be detailed in following sections.

Finally, since the multi-view approach requires outsourcing some
supplementary parameters, we will also need to analyze the
security/privacy of the communication protocol (privacy leakage in the
protocol, which complements the privacy analysis in output of the
protocol) in semi-honest model under the theory of secure multiparty
computation (SMC)~\cite{Yao86},~\cite{goldrich} (see
section~\ref{sec:secanal2}).

\section{The Building Blocks}
\label{sec:4}
In this section, we introduce the building blocks for our multi-view
mechanisms, namely, the \emph{iterative and reverse CryptoPAn},
\emph{partition-based} prefix preserving, and \emph{CryptoPAn with
  IP-collision (migration)}.

\subsection{Iterative and Reverse CryptoPAn}
As mentioned in section~\ref{subsec:Overview}, the multi-view approach
relies on iteratively applying a prefix preserving function $PP$ for
generating the seed view. Also, the analyst will invert such an
application of $PP$ in order to obtain the real view (among fake
views). Therefore, we first need to show how $PP$ can be iteratively
and reversely applied.

First, it is straightforward that $PP$ can be iteratively applied, and
the result also yields a valid prefix-preserving
function. Specifically, denote by $PP^{j}(a,K)$ ($j>1$) the
\emph{iterative} application of $PP$ on IP address $a$ using key $K$,
where $j$ is the number of iterations, called the \emph{index}. For
example, for an index of two, we have $PP^{2}(a,K)= PP(PP(a,K), K)$.
It can be easily verified that given any two IP addresses $a$ and $b$
sharing a k-bit prefix, $PP^{i}(a,K)$ and $PP^{i}(b,K)$ will always
result in two IP addresses that also share a k-bit prefix (i.e.,
$PP^i$ is prefix-preserving). More generally, the same also holds for
applying $PP$ under a sequence of indices and keys (for both IPs),
e.g., $PP^{i}(PP^{j}(a,K_0),K_1)$ and $PP^{i}(PP^{j}(b,K_0),K_1)$ will
also share k-bit prefix. Finally, for a set of IP addresses
$\mathcal{S}$, iterative $PP$ using a single key $K$ satisfies the following
associative property:

\begin{eqnarray}
& \hspace{-.38cm}  \forall \mathcal{S},K, \ \ \text{and} \ \ i, j \in \mathbb{Z}
 \; (integers): \ \ PP^{i}\big(PP^{j}(\mathcal{S},K), K\big)\nonumber \\
&   = PP^{j}\big(PP^{i} (\mathcal{S},K),
  K\big)=PP^{(i+j)} \big(\mathcal{S},K\big)
\label{eq:keyproperty}
\end{eqnarray}

On the other hand, when a negative number is used as the index, we
have a \emph{reverse} iterative CryptPAn function ($RPP$ for short),
as formally characterized in
Theorem~\ref{thm:reverese-prefic-preserving} (the proof is in
Appendix~\ref{subsec:proof1}).
\begin{thm}
\label{thm:reverese-prefic-preserving}
Given IP addresses $a=a_1a_2$ $\cdots a_{32}$ and $b=PP(a,K)=b_1b_2 \cdots b_{32}$, the function $RPP(\cdot):\{0,1\}^{32} \rightarrow \{0,1\}^{32}$ defined as
	\begin{equation}
	\label{RPP}
	\begin{split}
	&RPP(b,K)=c=c_1c_2 \cdots c_{32}\\
	& \text{where} \ c_{i}=b_{i} \oplus f_{i-1}(c_1 \cdots c_{i-1})
	\end{split}
	\end{equation}
is the inverse of the $PP$ function given in Equation~\ref{CryptoPAn}, i.e., $c=a$.
\end{thm}
\subsection{Partition-based Prefix Preserving}
\label{sec:pbpp}
As mentioned in section~\ref{subsec:Overview}, the central idea of the
multi-view approach is to divide the trace into partitions (Step $2$),
and then anonymize those partitions iteratively, but for different
number of iterations (Step $3$). In this subsection, we discuss this
concept.

Given $\mathcal{S}$ as a set of $n$ IP addresses, we may divide
$\mathcal{S}$ into partitions in various ways, e.g., forming
equal-sized partitions after sorting $\mathcal{S}$ based on either the
IP addresses or corresponding timestamps. The partitioning scheme will
have a major impact on the privacy, and we will discuss two such
schemes in next section.

Once the trace is divided into partitions, we can then apply $PP$ on
each partition separately, denoted by $PP(P_{i},K)$ for the $i^{th}$
partition.  Specifically, given $\mathcal{S}$ divided as a set of $m$
partitions \{$P_1, P_2, \cdots, P_m$\}, we define a \emph{key vector}
$V = \begin{bmatrix} v_1 & v_2$ $& \cdots & v_m \end{bmatrix}$ where
each $v_i$ is a positive integer indicating the number of times $PP$
should be applied to $P_i$, namely, the \textit{key index} of $P_i$.
Given a cryptographic key $K$, we can then define the
\textit{partition-based} prefix preserving anonymization of
$\mathcal{S}$ as $PP(\mathcal{S}, V, K)= \big[PP^{v_1}(P_1,K),
  \ PP^{v_2}(P_2, K),$ $\ldots , PP^{v_{m}}(P_{m},K)\big]$.

We can easily extend the associative property in
Equation~\ref{eq:keyproperty} to this case as the following (which
will play an important role in designing our multi-view mechanisms in
next section).

\begin{equation}
\label{keyproperty1}
PP[PP(\mathcal{S}, V_1, K), V_2, K]=PP (\mathcal{S}, (V_1+V_2),K)
\end{equation}

\subsection{IP Migration: Introducing IP-Collision into CryptoPAn}
\label{sec:migration sec}
As mentioned in section~\ref{subsec:Overview}, once the analyst
(adversary) receives the seed view, he/she would generate many
indistinguishable views among which only one, the real view, will be
prefix preserving across all the partitions, while the other (fake)
views do not preserve prefixes across the partitions (Step
5). However, the design would have a potential flaw under a direct
application of CryptoPAn. Specifically, since the original CryptoPAn
design is collision resistant~\cite{PP}, the fact that similar
prefixes are only preserved in the real view across partitions would
allow an adversary to easily distinguish the real view from
others.
\begin{figure}[ht]
		\includegraphics[width=1\linewidth,viewport= 210 160 680 345,clip]{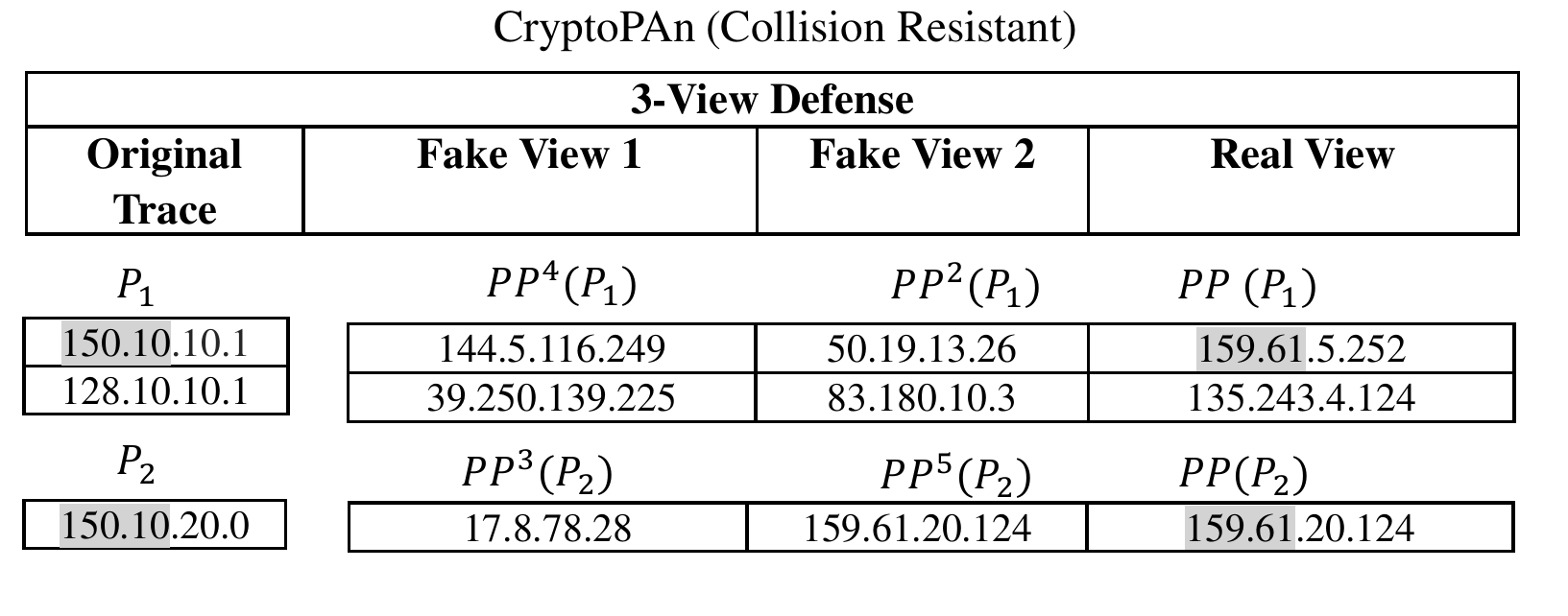}
		\caption{ An example showing only the real view contains shared prefixes (which allows it to be identified by adversaries)}
		\label{fig:migration0}
\end{figure}

\begin{exmp}
\label{ex:mig0}
Figure~\ref{fig:migration0} illustrates this flaw. The original trace
includes three different addresses and has been divided into two
partitions $P_1$ and $P_2$. As illustrated in the figure, the real
view is easily distinguishable from the two fake views as the shared
prefixes ($159.61$) between addresses in $P_1$ and $P_2$ only appear
in the real view. This is because, since the partitions in fake views have different rounds of PP
applied, and since the original CryptoPan design is collision resistant~\cite{PP}, the shared prefixes will no longer appear.
Hence, the adversary can easily distinguish the real view from others.
\end{exmp}

To address this issue, our idea is to create collisions between
different prefixes in fake views, such that adversaries cannot tell
whether the shared prefixes are due to prefix preserving in the real
view, or due to collisions in the fake views. However, due to the
collision resistance property of CryptoPAn~\cite{PP}, there is only a
negligible probability that different prefixes may become identical
even after applying different iterations of PP, as shown in the above
example. Therefore, our key idea of \emph{IP migration} is to first
replace the prefixes of all the IPs with common values (e.g., zeros),
and then fabricate new prefixes for them by applying different
iterations of PP. This IP migration process is designed to be
prefix-preserving (i.e,. any IPs sharing prefixes in the original
trace will still share the new prefixes), and to create collisions in
fake views since the addition of key indices during view generation
can easily collide. Next, we demonstrate this IP migration
technique in an example.

\begin{figure}[ht]
		\includegraphics[width=1\linewidth,viewport= 230 232 800 440,clip]{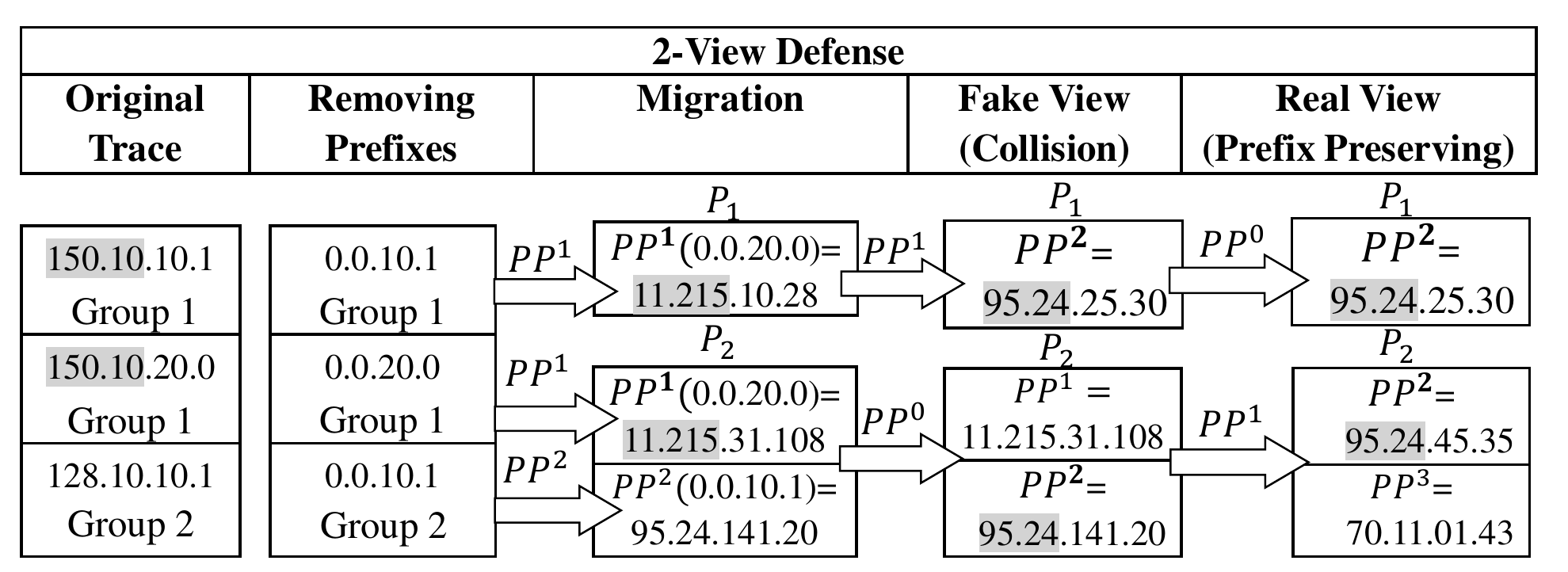}
		\caption{An example showing, by removing shared prefixes and fabricating them with the same rounds of PP, both fake view and real view may now contain fake or real shared prefixes (which makes them indistinguishable)}
		\label{fig:migration1}
\end{figure}
\begin{exmp}
\label{ex:mig1}
In Figure~\ref{fig:migration1}, the first stage shows the same
original trace as in Example~\ref{ex:mig0}. In the second stage, we
``remove'' the prefixes of all IPs and replace them with all zeros (by
xoring them with their own prefixes). Next, in the third stage, we
fabricate new prefixes by applying different iterations of $PP$ in a
prefix preserving manner, e.g., the first two IPs still sharing a
common prefix ($11.215$) different from that of the last IP.  However,
note that whether two IPs share the new prefixes only depends on their
key indices now, e.g., $1$ for first two IPs and $2$ for the last
IP. This is how we can create collisions in the next stage (the fake
view) where the first and last IPs coincidentally share the same
prefix $95.24$ due to their common key indices $2$ (however, note
these are the addition results of different key indices from the
migration stage and the view generation stage, respectively). Now, the
adversary will not be able to tell which of those views is real based
on the existence of shared prefixes.
\end{exmp}
We now formally define the migration function in the following.
\begin{defn}\textit{\textbf{Migration Function}}:
\label{def:mig}
Let $\mathcal{S}$ be a set of IP addresses consists of $d$ groups of IPs $S_1,S_2,\cdots,S_{d}$ with distinct prefixes $s_1,s_2,\cdots,s_{d}$ respectively, and $K$ be a random CryptoPAn key. Migration function $M:$ $\mathcal{S} \times \mathcal{C}(\text{set of positive integers}) \rightarrow \mathcal{S}^*$ is defined as
\begin{align}
\label{migfunc}
&\mathcal{S} ^*=M(\mathcal{S})=  \{S^*_{i}|\forall i\in\{1,2,\cdots,d\} \} \nonumber\\
& \text{where   } S^*_{i}=\{PP^{c_{i}}(s_{i}\oplus a_{j}, K), \forall a_{j} \in S_{i}\}
\end{align}
where $\mathcal{C}=PRNG(d,d)=\{c_1,c_2,$ $\cdots,c_{d}\}$ is the set of $d$ non-repeating random key indices generated between $[1,d]$ using a cryptographically secure pseudo random number generator.
\end{defn}

\section{$\epsilon$-Indistinguishable Multi-view Mechanisms}
\label{sec:instance}

We first present a multi-view mechanism based on IP partitioning in
Section~\ref{subsec:methd}. We then propose a more refined scheme
based on distinct IP partitioning with key vector generator in
Section~\ref{scheme2}.

\subsection{Scheme~\RN{1}: IP-based Partitioning Approach}
\label{subsec:methd}
To realize the main ideas of multi-view anonymization, as introduced
in Section~\ref{subsec:Overview}, we need to design concrete schemes
for each step in Figure~\ref{fig:2}.  The key idea of our first scheme
is the following. We divide the original trace in such a way that all
the IPs sharing prefixes will always be placed in the same
partition. This will prevent the attack described in
Section~\ref{sec:migration sec}, i.e., identifying the real view by
observing shared prefixes across different partitions.  As we will
detail in Section~\ref{sec:discussionscheme1}, this scheme can achieve
perfect indistinguishability without the need for IP migration
(introduced in Section~\ref{sec:migration sec}), although it has its
limitations which will be addressed in our second scheme. Both schemes
are depicted in Figure~\ref{fig:detailed-approach} and detailed below.

\begin{figure*}[ht]
		\includegraphics[width=0.68\linewidth,  viewport= 142 40 802 653,clip]{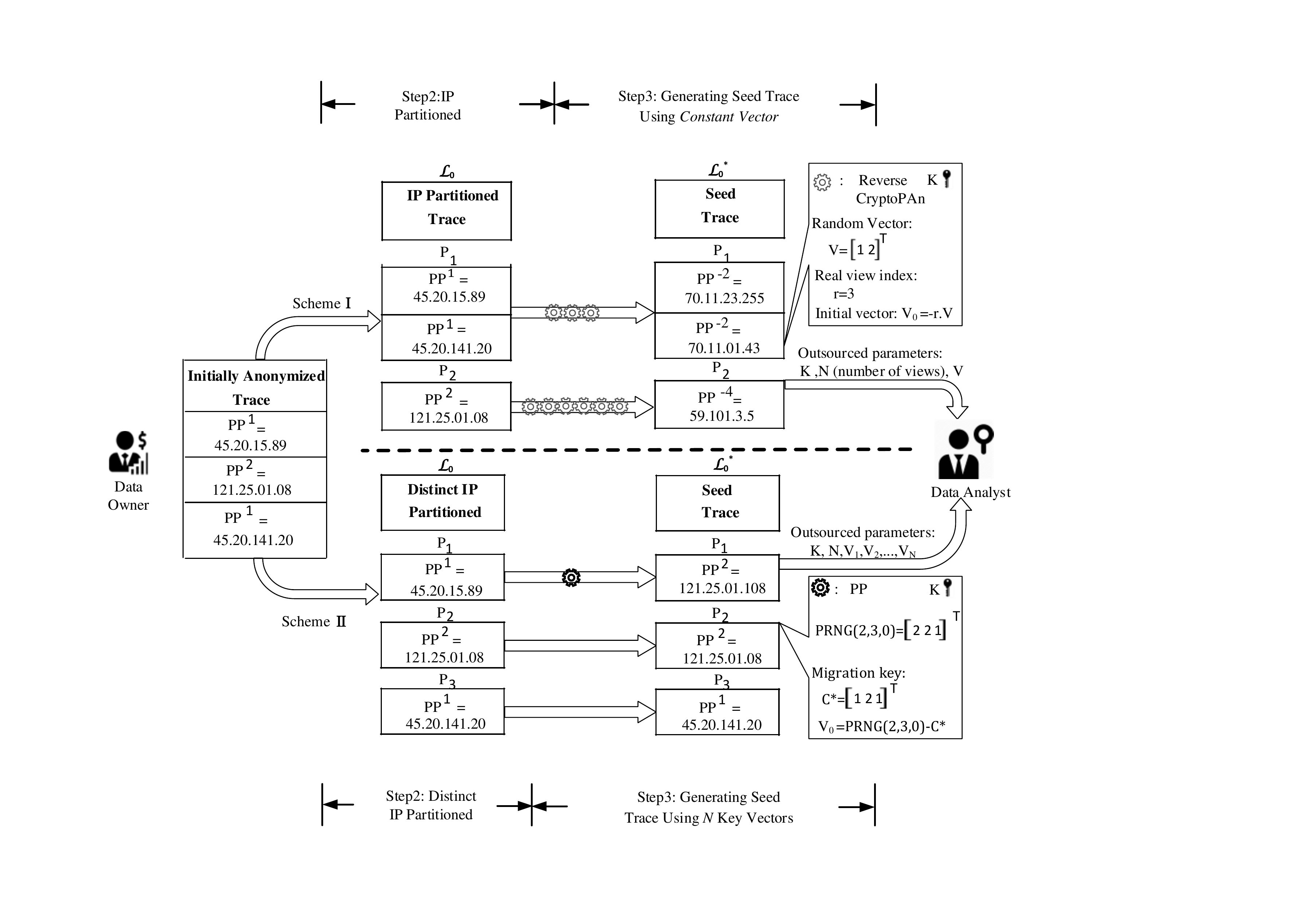}
		\caption{An example of a trace which undergoes multi-view schemes~\RN{1}, \RN{2}}
		\label{fig:detailed-approach}
\end{figure*}

Specifically, our first scheme includes three main steps: privacy
preservation (Section~\ref{sec:netwok-log-anonymization}), utility
realization (Section~\ref{sec:network-log-analysis}), and analysis
report extraction (Section~\ref{sec:analysis-report-extrcation}).

\subsubsection{Privacy Preservation (Data Owner)}
\label{sec:netwok-log-anonymization}
The data owner performs a set of actions to generate the seed trace
$\mathcal{L}^*_0$ together with some parameters to be sent to the
analyst for generating different views. These actions are summarized
in Algorithm~\ref{alg:data-owner-actions}, and detailed in the
following.

\begin{itemize}
\item{\emph{Applying CryptoPAn using $K_0$:}} First, the data owner generates
  two independent keys, namely $K_0$ (key used for initial
  anonymization, which never leaves the data owner) and $K$ (key used
  for later anonymization steps). The data owner then generates the
  initially anonymized trace $\mathcal{L}_0 $=$PP
  (\mathcal{L},K_0)$. This step is designed to prevent the adversary
  from simulating the scheme, e.g., using a brute-force attack to
  revert the seed trace back to the original trace in which he/she can
  recognize some original IPs. The leftmost block in
  Figure~\ref{fig:detailed-approach} shows an example of the initially
  anonymized trace.

\item{\emph{Trace partitioning based on IP-value:}} The initially anonymized
  trace is partitioned based on IP values. Specifically, let
  $\mathcal{S}$ be the set of IP addresses in $\mathcal{L}_0$
  consisting of $d$ groups of IPs $S_1,S_2,\cdots,S_{d}$ with distinct
  prefixes $s_1,s_2,\cdots,s_{d}$, respectively; we divide
  $\mathcal{L}_0$ to $d$ partitions, each of which is the collection
  of all records containing one of these groups. For example, the
  upper part of Figure~\ref{fig:detailed-approach} depicts how our
  first scheme works. The set of three IPs are divided into two
  partitions where $P_1$ includes both IPs sharing the same prefix,
  $45.20.15.89$ and $45.20.141.20$, whereas the last IP $121.25.01.08$
  goes to $P_2$ since it does not share a prefix with others.

\item{\emph{Seed trace creation:}} The data owner in this step generates the
  seed trace using a $d$-size (recall that $d$ is the number of
  partitions) random key vector.
\begin{itemize}
\item{\emph{Generating a random key vector:}} The data owner generates a
  random vector $V$ of size $d$ using a cryptographically secure
  pseudo random number generator $PRNG(d$ $,d) $ (which generates a set
  of $d$ non-repeating random numbers between $[1,d]$). This vector
  $V$ and the key $K$ will later be used by the analyst to generate
  different views from the seed trace.  For example, in
  Figure~\ref{fig:detailed-approach}, for the two partitions,
  $V=\begin{bmatrix} 1&2 \end{bmatrix}$ is generated. Finally, the
  data owner chooses the total number of views $N$ to be generated
  later by the analyst, based on his/her requirement about privacy and
  computational overhead, since a larger $N$ will mean more
  computation by both the data owner and analyst but also more privacy
  (more real view candidates will be generated which we will further study this through experiments later).
\item{\emph{Generating a seed trace key vector and a seed trace:}} The
  data owner picks a random number $r \in [1,N]$ and then computes
  $V_0=-r \cdot V$ as the key vector of seed trace. Next, the data
  owner generates the seed trace as $\mathcal{L}^*_0 = PP
  (\mathcal{L}_0,V_0, K)$. This ensures, after the analysts applies
  exactly $r$ iterations of $V$ on the seed trace, he/she would get
  $\mathcal{L}_0$ back (while not being aware of this fact since
  he/she does not know $r$). For example, in
  Figure~\ref{fig:detailed-approach}, $r=3$ and $V_0=\begin{bmatrix}
  -3 & -6 \end{bmatrix}$. We can easily verify that, if the analyst
  applies the indices in $V$ on the seed trace three times, the
  outcome will be exactly $\mathcal{L}_0$ (the real view). This can be
  more formally stated as follows (the $r^{th}$ view
  $\mathcal{L}_r^*$  is actually the real
  view).
	\begin{equation*}
	\begin{split}
	\mathcal{L}^*_r &=  PP(\mathcal{L}_0^*,r \cdot V,K), ~~~~~\text{using~\eqref{keyproperty1}}\\
	&= PP(\mathcal{L}_0,V_0+ r \cdot V, K), ~~~~~\text{using~\eqref{keyproperty1}}\\
	&= PP(\mathcal{L}_0,-r \cdot V + r\cdot V,K)\\
    &= \mathcal{L}_0
	\end{split}	
	\end{equation*}	
\end{itemize}

\item{Outsourcing:} Finally, the data owner outsources $\mathcal{L}^*_0$, $V$, $N$ and $K$ to the analyst.

\end{itemize}

\subsubsection{Network Trace Analysis (Analyst)}
\label{sec:network-log-analysis}
The analyst generates the $N$ views requested by the data owner, which
is summarized in Algorithm~\ref{alg:analyst-actions} in
Appendix~\ref{algs} and formalized below.
\begin{equation}
\begin{split}
&\mathcal{L}^*_0, ~ \text{~is the seed view} \\
&\mathcal{L}^*_i  = PP(\mathcal{L}^*_{i-1},V,K),~~~ i \in \{1,\ldots,N\}
\end{split}
\end{equation}
Since boundaries of partitions must be recognizable by the
analyst to allow him/her to generate the views, we modify the
time-stamp of the records that are on the boundaries of each partition
by changing the most significant digit of the time stamps which is
easy to verify and does not affect the analysis as it can be reverted
back to its original format by the analyst.  Next, the analyst
performs the requested analysis on all $N$ views and generates $N$
analysis reports $\Gamma_1, \Gamma_2, \cdots, \Gamma_N$.
\subsubsection{Analysis Report Extraction (Data Owner)}
\label{sec:analysis-report-extrcation}
The data owner is only interested in the analysis report that is
related to the real view, which we denote by $\Gamma_r$. To minimize
communication overhead, instead of requesting all the analysis reports
$\Gamma_i$ of the generated views, the data owner can fetch only the
one that is related to the real view $\Gamma_r$. He/she can employ the
\textit{oblivious random accesses memory} (ORAM)~\cite{oram} to do so
without revealing the information to the analyst (we will discuss
alternatives in Section~\ref{sec:related}).

\subsubsection{Security Analysis}
\label{sec:discussionscheme1}
We now analyze the level of indistinguishability provided by the
scheme. Recall the indistinguishability property defined in
Section~\ref{sec:model}; a multi-view mechanism is
$\epsilon$-indistinguishable if and only if
\begin{align*}
\exists \ \epsilon \geq 0, \  \textrm{s.t.} \ \forall i \in \{1,2, \cdots, N\} \Rightarrow \nonumber \\
 & \hspace{-5.1cm} e^{-\epsilon} \leq \frac{Pr(\text{view $i$ may be the real view})}{Pr(\text{view $r$ may be the real view})} \leq e^{\epsilon}
\end{align*}
The statement inside the probability is the adversary's decision on a
view, declaring it as fake or a \emph{real view candidate}, using
his/her $\mathcal{S}_{\alpha}$ knowledge. Moreover, we note that
generated views differ only in their IP values (fp-QI attributes are
similar for all the views). Hence, the adversary's decision can only
be based on the published set of IPs in each view through comparing
shared prefixes among those IP addresses which he/she already know
($\mathcal{S}_{\alpha}$). Accordingly, in the following, we define a
function to represent all the prefix relations for a set of IPs.

\begin{lem}
For two IP addresses $a$ and $b$, function $Q:\{0,1\}^{32} \times\{0,1\}^{32}\rightarrow \mathbb{N}$ returns the number of bits in the prefix shared between $a$ and $b$
\begin{align*}
Q(a, b)=31-\floor{log^{a \oplus b}_2}
\end{align*}
where $\floor{.}$ denotes the \emph{floor} function.
\end{lem}

\begin{defn}
	\label{dfn:4.2}
For a multiset of $n$ IP addresses $\mathcal{S}$, the \emph{Prefixes
  Indicator Set} (PIS) $\mathcal{R}(\mathcal{S})$ is defined as
follows.
\begin{align}
\label{eq:PIS}
 \mathcal{R}(\mathcal{S})=\{Q(a_i, a_j)| \  \forall a_{i},a_{j} \in\mathcal{S},  i,j \in \{1,2, \cdots, n\} \}
\end{align}
\end{defn}
Note that PIS remains unchanged when CryptoPAn is applied on $\mathcal{S}$, i.e., $\mathcal{R}(PP(\mathcal{S},K))=\mathcal{R}(\mathcal{S})$. In addition, since the multi-view solution keeps all the other attributes intact, the adversary can identify his/her pre-knowledge in each view and construct prefixes indicator sets out of them. Accordingly, we denote by $\mathcal{R}_{\alpha,i}$ the PIS constructed by the adversary in view $i$.
\begin{defn}
Let $\mathcal{R}_{\alpha}$ be the PIS for the adversary's knowledge, and $\mathcal{R}_{\alpha,i}$, $i \in \{1,\cdots,N\}$ be the PIS constructed by the adversary in view $i$. A multi-view solution then generates \emph{$\epsilon$-indistinguishable} views against an $\mathcal{S}_{\alpha}$ adversary if and only if
\begin{align*}
 \exists \ \epsilon \geq 0, \  \textrm{s.t.} \ \forall i \in \{1,2, \cdots, N\} \Rightarrow
\end{align*}
\begin{align}
\label{indis:eqn}
e^{-\epsilon} \leq \frac{Pr(\mathcal{R}_{\alpha,i}=\mathcal{R}_{\alpha})}{Pr(\mathcal{R}_{\alpha,r}=\mathcal{R}_{\alpha})} \leq e^{\epsilon}
\end{align}
\end{defn}
\begin{lem}
The indistinguishability property, defined in equation~\ref{indis:eqn} can be simplified to
\begin{align}
\label{eq:epsasl}
 & \exists \ \epsilon \geq 0, \  \textrm{s.t.} \ \forall i \in \{1,2, \cdots, N\} \Rightarrow \nonumber\\ \nonumber& Pr(\mathcal{R}_{\alpha,i}=\mathcal{R}_{\alpha}) \geq e^{-\epsilon} \\
\end{align}
\end{lem}
\begin{proof}
$Pr(\mathcal{R}_{\alpha,r}=\mathcal{R}_{\alpha})=1$ as view $r$ is the prefix preserving output. Moreover, $\forall \epsilon \geq 0$ we have $e^{\epsilon}\geq 1$.
\end{proof}
From the above, we only need to show
$\mathcal{R}_{\alpha,i}=\mathcal{R}_{\alpha}$ (each generated view $i$
is a real view candidate).

\begin{thm}
\label{lem:lem3}
Scheme~\RN{1} satisfies equation~\ref{eq:epsasl} with $\epsilon=0$.
\end{thm}
\begin{proof}
Scheme~\RN{1} divides the trace into $d$ (number of prefix groups)
partitions containing all the records that have similar
prefixes. Hence, for any partition $P_i$ ($1\leq i\leq d$), any two IP
addresses $a$ and $b$ inside $P_i$, and for any $m,n \leq N$, we have
$\mathcal{R}_{m}(a,b)=\mathcal{R}_{n}(a,b)$ because $a$ and $b$ are
always assigned with equal key indices. Moreover, for any two IP
addresses $a$ and $b$ in any two different partitions and any $m,n
\leq N$, we have $\mathcal{R}_{m}(a,b)=\mathcal{R}_{n}(a,b)=0$ since
they do not share any prefixes.
\end{proof}

The above discussions show that scheme~\RN{1} produces perfectly
indistinguishable views ($\epsilon=0$). In fact, it is robust against
the attack explained in Section~\ref{sec:migration sec} and thus does
not required IP migration, because the partitioning algorithm already
prevents addresses with similar prefixes from going into different
partitions (the case in Figure~\ref{fig:migration0}). However,
although adversaries cannot identify the real view, they may choose to
live with this fact, and attack each partition inside any (fake or
real) view instead, using the same semantic attack as shown in
Figure~\ref{tab:motivating-example}. Note that our multi-view
approach is only designed to prevent attacks across different
partitions, and each partition itself is essentially still the output
of CryptoPAn and thus still inherits its weakness.

Fortunately, the multi-view approach gives us more flexibility in
designing specific schemes to further mitigate such a weakness of
CryptoPAn. We next present scheme~\RN{2} which sacrifices some
indistinguishability (in the sense of slightly less real view
candidates) to achieve better protected partitions.

\subsection{Scheme~\RN{2}: Multi-view Using $N$ Key Vectors}
\label{scheme2}

To address the limitation of our first scheme, we propose the next
scheme, which is different in terms of the initial anonymization step,
IP partitioning, and key vectors for view generation. The data owner's and the analyst's actions are summarized
in Algorithms~\ref{alg:data-owner-actions1},~\ref{alg:analyst-actions1}.

\subsubsection{Initial Anonymization with Migration}
First, to mitigate the attack on each partition, we must
relax the requirement that all shared prefixes go into the same
partition. However, as soon as we do so, the attack of identifying the
real view through prefixes shared across partitions, as demonstrated
in Section~\ref{sec:migration sec}, might become possible. Therefore,
we modify the first step of the multi-view approach (initial
anonymization) to enforce the IP migration
technique. Figure~\ref{fig:migration12222} demonstrates this. The
original trace is first anonymized with $K_0$, and then the anonymized
trace goes through the migration process, which replaces the two
different prefixes ($97.17$ and $75.91$) with different iterations
of $PP$, as discussed in Section~\ref{sec:migration sec}.

\begin{figure}[ht]
		\includegraphics[width=1.1\linewidth,viewport= 260 140
                  900 440,clip]{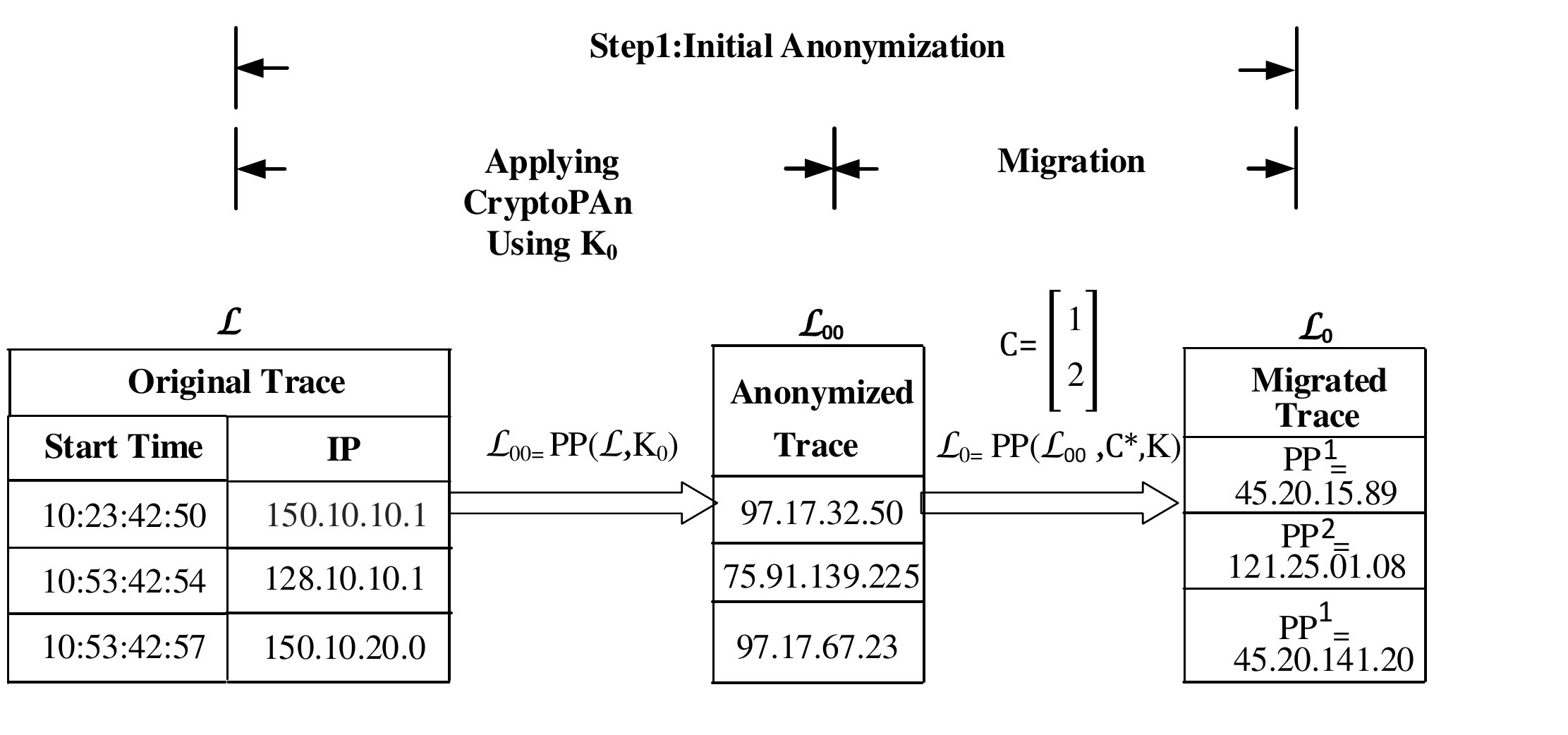}
		\caption{The updated initial anonymization (Step $1$ in  Figure~\ref{fig:2}) for enforcing migration}
		\label{fig:migration12222}
\end{figure}

\subsubsection{Distinct IP Partitioning and $N$ Key Vectors Generation}
For the scheme, we employ a special case of IP partitioning where each
partition includes exactly one distinct IP (i.e.,
the collection of all records containing the same IP). For example,
the trace shown in Figure~\ref{fig:detailed-approach} includes three
distinct IP addresses $150.10.10.1$,$128.10.10.1$, and
$150.10.20.0$. Therefore, the trace is divided into three
partitions. Next, the data owner will generate the seed view as in the
first scheme, although the key $V_0$ will be generated completely
differently, as detailed below.

Let $\mathcal{S}^*=\{S^*_1,S^*_2, \cdots, S^*_d\}$, be the set of IP
addresses after the migration step. Suppose $\mathcal{S}^*$ consists
of $D$ distinct IP addresses. We denote by $\mathcal{C}^*$ the
multiset of totally $D$ migration keys for those distinct IPs (in
contrast, the number of migration keys in $\mathcal{C}$ is equal to
the number of distinct prefixes, as discussed in
Section~\ref{sec:migration sec}). Also, let $PRNG(d,D,i)$ be the set
of $D$ random number generated between $[1,d]$ using a
cryptographically secure pseudo random number generator at iteration
$i^{th}$. The data owner will generate $N+1$ key vector $V_{i}$ as
follows.

 \begin{eqnarray}
 \label{eq:index generator}
& V_{i}= PRNG(d,D,i)-PRNG(d,D,i-1), \\
&\forall i \neq r \in [1,2 \cdots, N] \nonumber
 \end{eqnarray}
 and
 \begin{eqnarray}
 \label{eq:index generator0}
&V_{0}=PRNG(d,D,0)-\mathcal{C}^*\\
&V_{r}=\mathcal{C}^*-PRNG(d,D,r-1)\nonumber
 \end{eqnarray}

\begin{exmp}
In Figure~\ref{fig:detailed-approach121}, the migration and random
vectors are $ \mathcal{C}^*=[1\ 1 \ 2]$, $PRNG(2,3,0)=[1\ 2 \ 2]$,
$PRNG(2,3,1)=[1\ 2 \ 1]$, and $PRNG(2,3,2)=[2\ 2 \ 1]$,
respectively. The corresponding key vectors will be $V_0=[0\ 1 \ 0]$,
$V_1=[0\ 0 \ -1]$ and $V_2=[1\ 0 \ 0]$ where only $V_1$ and $V_2$ are
outsourced.
\end{exmp}

In this scheme, the analyst at each iteration $i$ generates a new set
of IP addresses $\mathcal{S}^*_{i}=\{S^i_1, S^i_2, \cdots, S^i_d\}$ by
randomly grouping all the distinct IP addresses into a set of $d$
prefix groups. In doing so, each new vector $V_i$ essentially cancels
out the effect of the previous vector $V_{i-1}$, and thus introduces a
new set of IP addresses $\mathcal{S}^*_{i}$ consisting of $d$ prefix
groups. Thus, it is straightforward to verify that the $r^{th}$
generated view will prefix preserving (the addresses are migrated back
to their groups using $\mathcal{C}^*$).

\begin{exmp}
  Figure~\ref{fig:detailed-approach121} shows that, in each iteration,
  a different set (but with an equal number of elements) of prefix
  groups will be generated. For example, in the seed view, IP
  addresses $150.10.20.0$ and $128.10.10.1$ are mapped to prefix group
  $11.215$.
\end{exmp}

\begin{figure}[!t]
		\includegraphics[width=1.21\linewidth,  viewport= 220 205 920 440,clip]{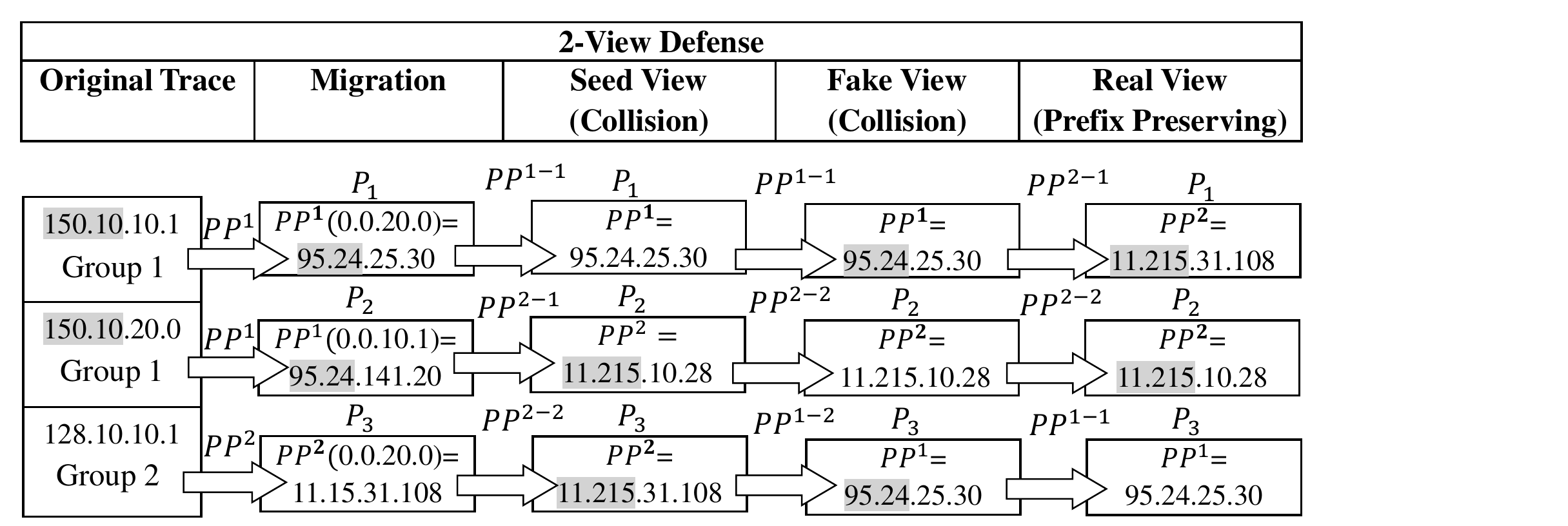}
		\caption{An example of three views generation under scheme~\RN{2}}
		\label{fig:detailed-approach121}
\end{figure}

\subsubsection{Indistinguishability Analysis}
\label{sec:analysis}
By placing each distinct IP in a partition, our second scheme is not
vulnerable to semantic attacks on each partition, since such a
partition contains no information about the prefix relationship among
different addresses. However, compared with scheme~\RN{1}, as we show
in the following, this scheme achieves a weaker level of
indistinguishability (higher $\epsilon$). Specifically, to verify the
indistinguishability of the scheme, we calculate
$Pr(\mathcal{R}_{\alpha}=\mathcal{R}_{\alpha,i})$ for scheme~\RN{2} in
the following. First, the number of all possible outcomes of grouping
$D$ IP addresses into $d$ groups with predefined cardinalities is:
\begin{align}
\hspace{-1cm}  N_{total}=\frac{D!}{|S_1|!|S_2|!\cdots |S_d|!}
\end{align}
where $|S_i|$ denotes the cardinality of group $i$. Also the number of all possible outcomes of grouping $D$ IP addresses into $d$ groups while still having  $\mathcal{R}_{\alpha,i}=\mathcal{R}_{\alpha}$ is:
\begin{align}
N_{\text{real view candidates}}=\frac{\alpha! \ (D-\alpha)! \ \sum^{ \binom{d}{\alpha}}_{i=1}  \big( \Pi^{\alpha}_{i=1}|S_{a_i}| \big)}{|S_1|!|S_2|!\cdots |S_d|!}
\end{align}
for some $a_i \in \{1,2, \cdots, d\}$. This equation gives the number of outcomes when a specific set of $\alpha$ IP addresses ($\mathcal{S}_{\alpha}$) are distributed into $\alpha$ different groups and hence keeping $\mathcal{R}_{\alpha,i}=\mathcal{R}_{\alpha}$ (i.e., the adversary cannot identify collision). Note that term $\sum^{ \binom{d}{\alpha}}_{i=1}  \big( \Pi^{\alpha}_{i=1}|S_{a_i}| \big)$ is all the combinations of choosing this $\alpha$ groups for the numerator to model all the $(|S_{a_i}|-1)!$ combinations. Finally, we have
\begin{align*}
 \forall i \leq N: \ \ Pr(\mathcal{R}_{\alpha,i}=\mathcal{R}_{\alpha})=\frac{N_{\text{real view candidates}}}{N_{total}}=
 \end{align*}
 \begin{align}
 \label{eqn:epsin}
 \mathcal{A}=\frac{\alpha! \ \sum^{ \binom{d}{\alpha}}_{i=1}  \big( \Pi^{\alpha}_{i=1}|S_{a_i}| \big)}{\Pi^{\alpha-1}_{i=0}(D-i)} \geq e^{-\epsilon}
\end{align}
Thus, to ensure the $\epsilon$-indistinguishability, the data owner
needs to satisfy the expression in equation~\ref{eqn:epsin} which is a
relationship between the number of distinct IP addresses, the number of groups, the cardinality of the groups in the trace and the adversary's knowledge.
\begin{thm}
\label{th:5.3}
The indistinguishability parameter $\epsilon$ of the generated views in scheme~\RN{2} is lower-bounded by
\begin{equation}
\label{eqn:optimepsin}
\ln\big[\frac{D^{\alpha}}{ d^{\alpha}} \cdot \Pi^{\alpha-1}_{i=0} \frac{(d-i) }{(D-i)} \big]
\end{equation}
\begin{proof}
Let $b_1, b_2, \cdots, b_n$ be positive real numbers, and for $k = 1, 2, \cdots, n$ define the averages $M_k$ as follows:
\begin{align}
M_k=\frac{\sum \limits_{1\leq i_1\leq i_2 \leq \cdots\leq i_{k} \leq n}b_{i_1}b_{i_1} \cdots b_{i_k}}{\binom{n}{k}}
\end{align}
By Maclaurin's inequality~\cite{mclauren}, which is the following chain of inequalities:
\begin{align}
M_1\geq \sqrt[2]{M_2} \geq \sqrt[3]{M_3} \geq \cdots \geq \sqrt[n]{M_n}
\end{align}
where $M_1= \frac{\sum \limits^{n}_{i=1} b_i}{n}$,
 we have
\begin{equation*}
 \mathcal{A}= \frac{\alpha! \binom{d}{\alpha} M_{\alpha}}{\Pi^{\alpha-1}_{i=0}(D-i)}\leq \frac{\Pi^{\alpha-1}_{i=0} (d-i) (M_1)^{\alpha}}{\Pi^{\alpha-1}_{i=0}(D-i)}
\end{equation*}
and since $M_1= \frac{\sum \limits^{n}_{i=1} |S_i|}{n}=\frac{D}{d}$, we have
\begin{equation*}
A \leq  \frac{D^{\alpha}}{ d^{\alpha}} \cdot \Pi^{\alpha-1}_{i=0} \frac{(d-i) }{(D-i)}
\end{equation*}
\end{proof}
\end{thm}
Figure~\ref{fig:detailed-approach1}(a) shows how the lower-bound in
Equation~\ref{eqn:optimepsin} changes with respect to different values
of fraction $d/D$ and also the adversary's knowledge. As it is
expected, stronger adversaries have more power to weaken the scheme
which results in increasing $\epsilon$ or increasing the chance of
identifying the real view. Moreover, as it is illustrated in the
figure, when fraction $d/D$ grows, $\epsilon$ tends to converge to
very small values. Hence, to decrease $\epsilon$, the data owner may
increase $d/D \in [0,1]$ by grouping addresses based on a bigger
number of bits in their prefixes, e.g., a certain combination of 3
octets would be considered as a prefix instead of one or two. Another
solution could be aggregating the original trace with some other
traces for which the cardinalities of each prefix group are small. We
study this effect in our experiments in Section~\ref{sec:experiments}
where we illustrate the concept especially in
Figures~\ref{fig:leakage23}, \ref{fig:time-computation1}.

Finally, Figure~\ref{fig:detailed-approach1}(b) shows how variance of
the cardinalities affects the indistinguishability for a set of fixed
parameters $d$, $D$, $\alpha$. In fact, when the cardinalities of the
prefix groups are close (small $\sigma$), $\mathcal{A}$ grows to meet
the lower-bound in Theorem~\ref{th:5.3}. Hence, from the data owner
perspective, a trace with a lower variance of cardinalities and a
bigger fraction $d/D$ has a better chance of misleading adversaries
who wants to identify the real view.
\begin{figure}[!t]
		\includegraphics[width=1.16\linewidth,  viewport= 10 205 920 620,clip]{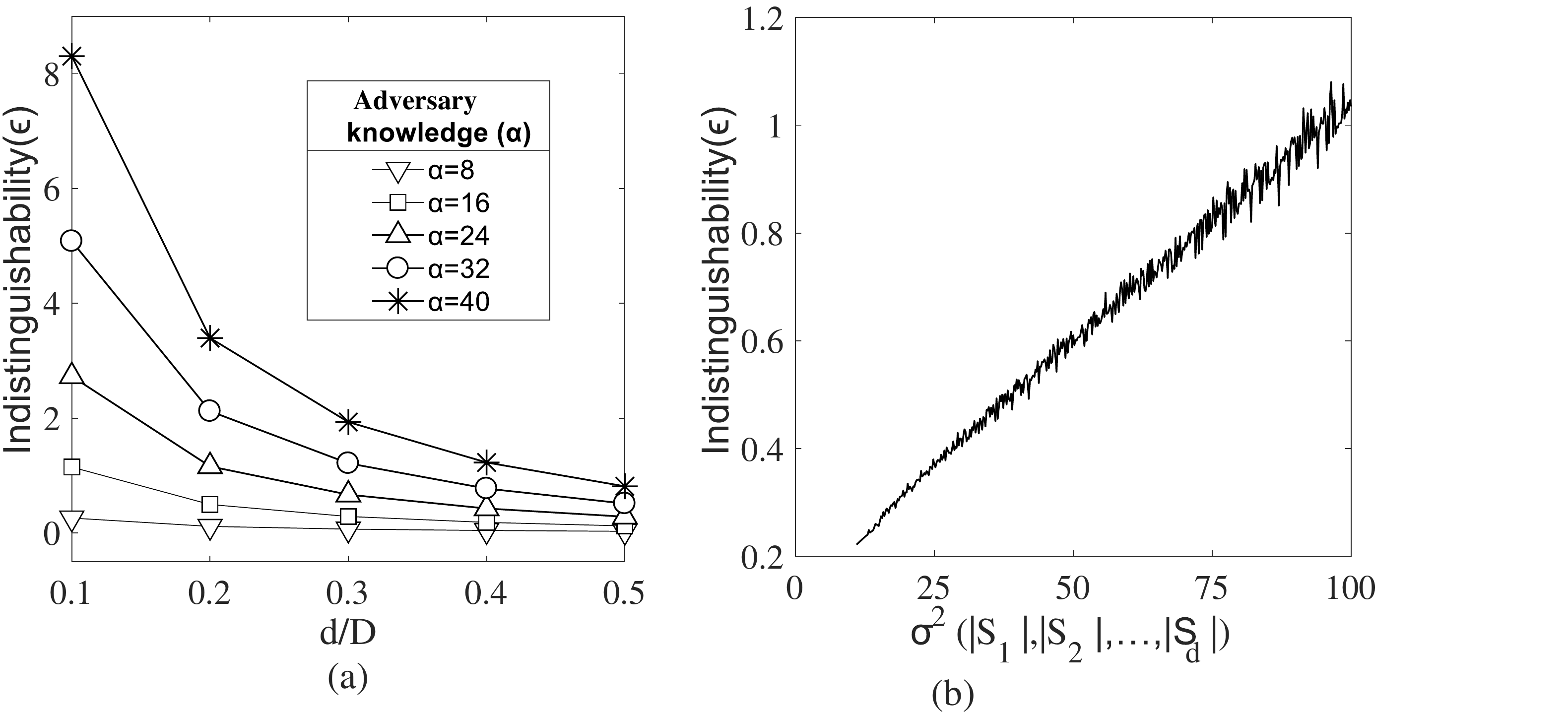}
		\caption{(a) The trend of bound~\ref{eqn:optimepsin} for $\epsilon$ when adversary's knowledge varies. (b) The trend of exact value of $\epsilon$ in equation~\ref{eqn:epsin} for $\alpha=8$, $d/D=0.1$ and when variance of cardinalities varies}
		\label{fig:detailed-approach1}
\end{figure}
\subsubsection{Security of the communication protocol}
\label{sec:secanal2}

We now analyze the security/privacy of our communication protocol in semi-honest model under the theory of secure multiparty computation (SMC)~\cite{Yao86}, \cite{goldrich}.
\begin{lem}
Scheme~\RN{2} only reveals the CryptoPan Key $K$ and the seed trace $\mathcal{L}^*_0$ in semi-honest model.
\end{lem}

\begin{proof}
Recall that our communication protocol only involves one-round communication between two parties (data owner to data analyst). We then only need to examine the data analyst's view (messages received from the protocol), which includes (1) $N$: number of views to be generated, (2) $K$: the outsourced key, (3) $\mathcal{L}^*_0$: the seed trace, and (4) $V_1,V_2,\cdots,V_N$: the key vectors. As we discuss in section~\ref{sec:analysis}, the probability of identifying the real view by the adversary using all provided information (key and vectors) depends on the adversary knowledge and the trace itself which clearly implies that such ``leakage'' is trivial.

Indeed, each of $N$ and $V_1, V_2, \dots, V_N$ can be simulated by
generating a single random number from a uniform random distribution (which proves that they are not leakage in the protocol). Specifically,  the number of generated views $N$ is integer which is bounded by $N_0$, where $N_0$ is the maximum number of views the data owner can afford and all the entries in $V_1,V_2,\cdots ,V_N$ are in $[-d,d]$ where $d$ is the number of groups. First, given integer $0<N\in N^*$, the probability that $N$ is simulated in the domain would be $Pr[Simulator=N]=\frac{1}{N^*}$. Then, $N$ can be simulated in polynomial time (based on the knowledge data analyst already knew, i.e., his/her input and/or output of the protocol). Similarly, all the random entires in $V_1,V_2,\cdots,V_N$ can also be simulated in polynomial time using a similar simulator (only changing the bound). Thus, the protocol only reveals the outsourced key $K$ and the seed trace $\mathcal{L}^*_0$ in semi-honest model.
\end{proof}
Note that outsourcing the $\mathcal{L}_0^*$ and the outsourced key are trivial leakage. The outsourced key can be considered as a public key and leakage of $\mathcal{L}_0^*$ which is considered as the output of the protocol was studied earlier.
Finally, we study the \emph{setup leakage} and show that the adversary cannot exploit outsourced parameters to increase $\epsilon$ (i.e., decrease the number of real view candidates) by building his/her own key vector.
\begin{lem}(proof in Appendix~\ref{prooflem1})
\label{lem:jj}
For an $\mathcal{S}_{\alpha}$ adversary, who wants to obtain the least number of real view candidates, if condition $(2d-2)^D>N$ holds, the best approach is to follow scheme~\RN{2}, (scheme~\RN{2} returns the least number of real view candidates).
\end{lem}
\subsection{Discussion}
\label{ffaa}
In this section, we discuss various aspects and limitations of our approach.

\begin{enumerate}
\item \textbf{Application to EDB:} We believe the
  multi-view solution may be applicable to other related areas. For
  instance, processing on encrypted databases (EDB) has a rich
  literature including searchable symmetric encryption
  (SSE)~\cite{sse1}, ~\cite{sse2}, fully-homomorphic encryption (FHE)
  \cite{FHE}, oblivious RAMs (ORAM)~\cite{goldrich}, functional
  encryption~\cite{boneh}, and property preserving encryption
  (PPE)~\cite{Bellare},~\cite{Boldyreva}. All these approaches achieve
  different trade-offs between protection (security), utility (query
  expressiveness), and computational
  efficiency~\cite{naveed}. Extending and applying the multi-view
  approach in those areas could lead to interesting future directions.

\item\textbf{Comparing the Two Schemes:} As we discussed in the two schemes,
  scheme~\RN{1} achieves a better indistinguishability but less
  protected partitions in each view. Figure~\ref{fig:a233} compares
  the relative effectiveness of the two schemes on a real trace under
  $40\%$ adversary knowledge. In particular, Figure~\ref{fig:a233}(a)
  ,(b) demonstrate the fact that despite the lower number of real view
  candidates in scheme~\RN{2} compared with scheme~\RN{1} ($30$ vs
  $160$ out of $160$), the end result of the leakage in scheme~\RN{2}
  is much more appealing ($3\%$vs $35\%$). Therefore, our experimental
  section has mainly focused on scheme~\RN{2}.

\item\textbf{Choosing the Number of Views $N$:} The number of views
  $N$ is an important parameter of our approach that determines both
  the privacy and computational overhead. The data owner could choose
  this value based on the level of trust on the analysts and the
  amount of computational overhead that can be afforded. Specifically,
  as it is implied by Equation~\ref{eq:epsasl} and demonstrated by our
  experimental results in section~\ref{sec:experiments}, the number of
  real view candidates is approximately $e^{-\epsilon}\cdot N$. The
  data owner should first estimate the adversary's background
  knowledge $\alpha$ (number of prefixes known to the adversary) and
  then calculate $\epsilon$ either using Equation~\ref{eqn:epsin} or
  (approximately) using Equation~\ref{eqn:optimepsin}. As it is
  demonstrated in Figures~\ref{fig:detailed-approach1}(a)
  and~\ref{figepsasl4}(b), a bigger $\alpha$ results in weaker
  indistinguishability and demands a larger number of views to be
  generated. An alternative solution is to increase the number of
  prefix groups ($D$) by sacrificing some prefix relations among IPs,
  e.g., grouping them based on first $3$ octets.

\item\textbf{Utility:} The main advantage of the multi-view approach
  is it can preserve the data utility while protecting privacy. In
  particular, we have shown that the data owner can receive an analysis
  report based on the real view ($\Gamma_r$) which is
  prefix-preserving over the entire trace. This is more accurate than
  the obfuscated (through bucketization and suppression) or perturbed
  (through adding noise and aggregation) approaches.  Specifically, in
  case of a security breach, the data owner can easily compute
  $\mathcal{L}_{r}$ (migration output) to find the mapped IP addresses
  corresponding to each original address. Then the data owner applies
  necessary security policies to the IP addresses that are reported
  violating some policies in $\Gamma_r$. A limitation of our work is
  it only preserve the prefix of IPs, and a potential future direction
  is to apply our approach to other property-preserving encryption
  methods such that other properties may be preserved similarly.

\item\textbf{Communicational/Computational Cost:} One of our
  contributions in this paper is to minimize the communication
  overhead by only outsourcing one (seed) view and some supplementary
  parameters. This is especially critical for large scale network data
  like network traces from the major ISPs. On the other hand, one of
  the key challenges to the multi-view approach is that it requires
  $N$ times computation for both generating the views and
  analysis.

  Our experiments in Figure~\ref{fig:time-computation1}
  shows that generating $160$ views for a trace of $1milion$ packets
  takes approximately $4$ minutes and we describe analytic complexity
  results in Tables~\ref{relatedss1} and ~\ref{relatedss2}. We note
  that the practicality of $N$ times computation will mainly depends
  on the type of analysis, and certainly may become impractical for
  some analyses under large $N$. How to enable analysts to more
  efficiently conduct analysis tasks based on multiple views through
  techniques like caching is an interesting future direction. Another
  direction is to devise more accurate measures for the data owner to
  more precisely determine the number of views required to reach a
  certain level of privacy requirement.
 \end{enumerate}

\section{Experiments}

This section evaluates our multi-view scheme through experiments with
real data.

\label{sec:experiments}
\subsection{Setup}
To validate our multi-view anonymization approach, we use a set of
network traces collected by a real ISP. We focus on attributes
$Timestamp$, $IP address$, and $Packet Size$ in our experiments, and
the meta-data are summarized in the table in Figure~\ref{tab:data-specs}(a).
\begin{figure}[ht]
	\begin{center}
		\includegraphics[width=1.1\textwidth, viewport= 130 130  1620 430, clip]{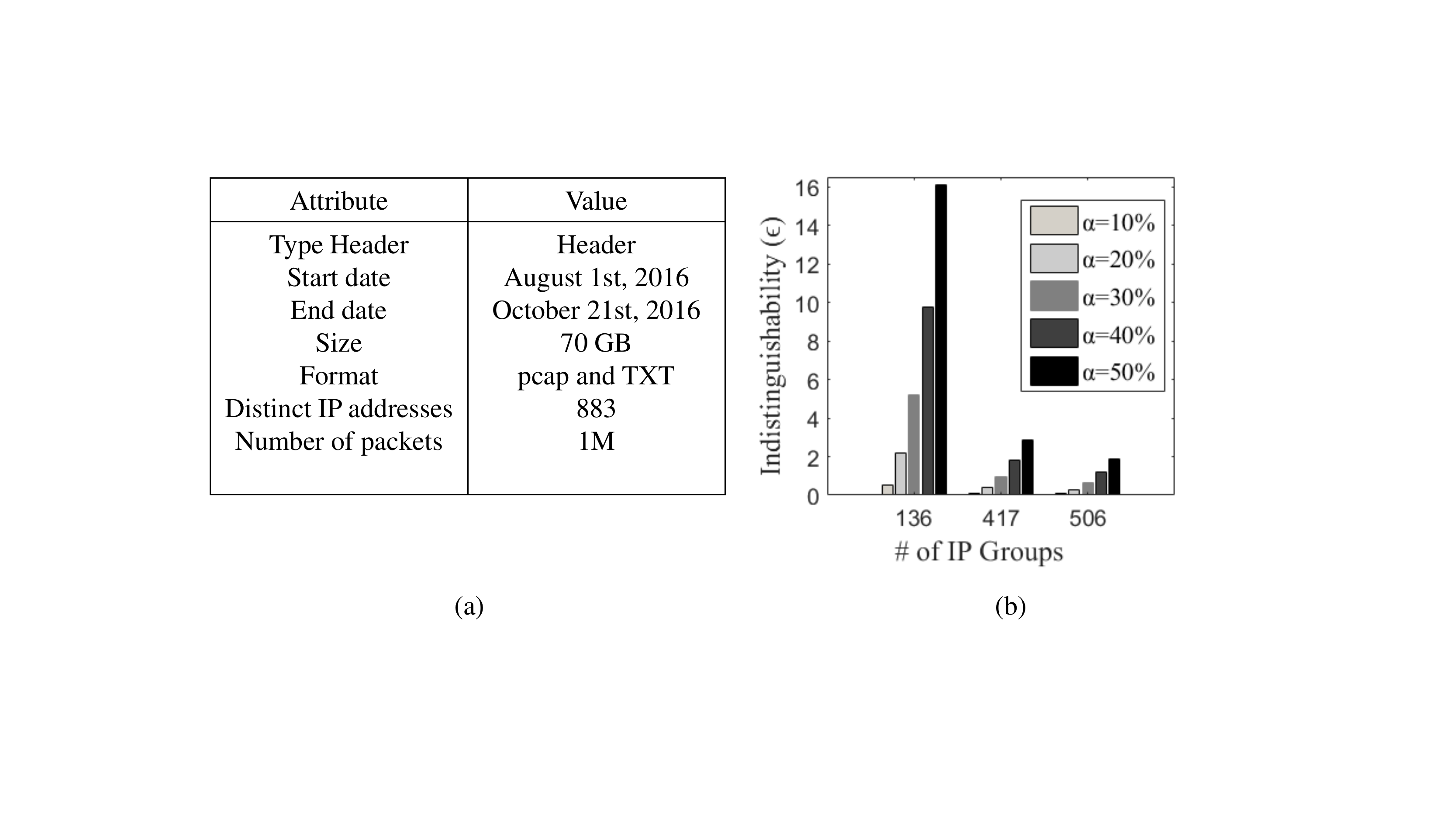}
		\caption{(a) Metadata of the collected traces (b) $\epsilon$ for different number of prefix groups and
different adversary knowledges}
		\label{tab:data-specs}
        \label{figepsasl4}
	\end{center}
\end{figure}
In order to measure the security of the proposed approach, we
implement the frequency analysis attack~\cite{naveed},~\cite{brekene1}. This attack can compromise individual
addresses protected by existing prefix-preserving anonymization in
multi-linear time~\cite{brekene1}. We stress that in the setting of EDBs (encrypted database systems), an attack is successful if it recovers even partial information about a single cell of the DB~\cite{naveed}. Accordingly, we define the information leakage metric to evaluate the effectiveness of our solution against the adversary's semantic attacks. Several measures have been proposed in literature~\cite{PP,Ribeiro} to evaluate the impact of semantic attacks. Motivated by~\cite{PP}, we model the information leakage (number of matches) as the number of records/packets, their original IP addresses are known by the adversary either fully or partially. More formally,\\
\textbf{\textit{\textbf{Information leakage metric}}}~\cite{PP}:
\label{def:infoleak}
We measure $F_{i}$ defined as the total number of addresses that has at least $i$ most significant bits known, where $i \in \{1,2,\cdots,32\}$.

To model adversarial knowledge, we
define a set of prefixes to be known by the adversary ranging from
$10\%$ up to $100\%$ of all the prefixes in the trace. This knowledge
is stored in a two dimensional vector that includes $\alpha$ different
addresses and their key indexes.
\begin{figure*}[ht]
	\begin{center}
		\includegraphics[width=0.9\textwidth, viewport= -5 3  1100 660, clip]{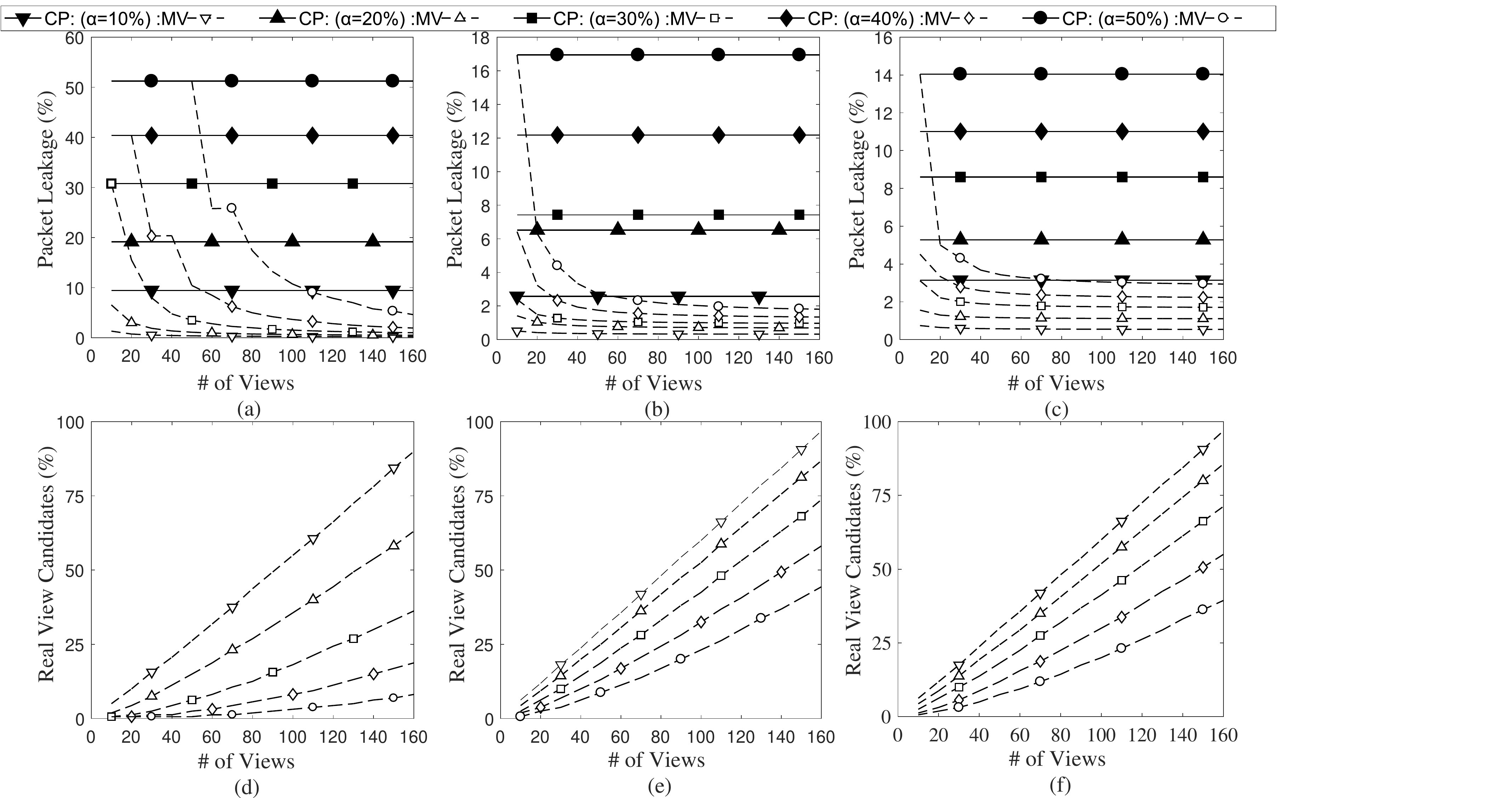}
		\caption{Percentage of the compromised packets (out of 1M) and number of real view candidates when number of views and the adversary knowledge vary and for the three different cases (1) Figures (a),(d) (2) Figures (b),(e) (3) Figures (c),(f) where legends marked by \emph{CP} denote the CryptoPAn result whereas those marked by \emph{MV} denote the multi-view results}
		\label{fig:leakage23}
	\end{center}
\end{figure*}
\begin{figure}[ht]
		\includegraphics[width=0.52\textwidth, viewport= -80 390  430 640, clip]{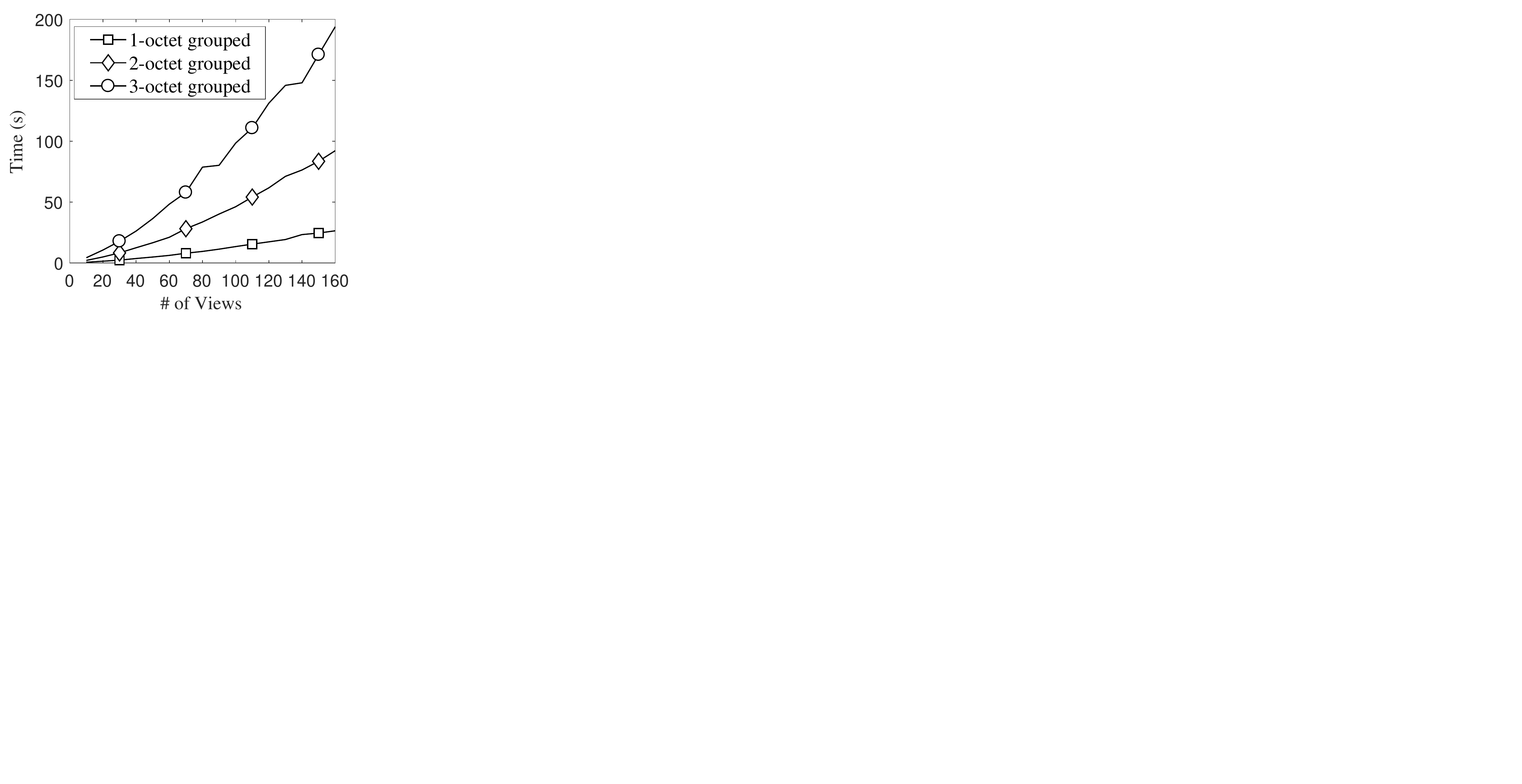}
		\caption{Computation time obtained by our anonymization approach for different prefix grouping cases}
		\label{fig:time-computation1}
\end{figure}
Next, using our multi-view scheme, we generate all the $N$ views.
However, before we apply the frequency analysis attack, we simulate
how an adversary may eliminate some fake views from further
consideration as follows. For each view, we check if two addresses
from the adversary's knowledge set with different prefixes now share
prefixes in that view.  If we find such a match in the key indices,
the corresponding view will be discarded from the set of the real view
candidates and will not be considered in our experiments since the
adversary would know it is a fake view.

We validate the effectiveness of our scheme by showing the number of
real view candidates and the percentage of the packets in the trace
that are compromised (i.e., the percentage of IP packets whose
addresses have at least eight most significant bits known). Each
experiment is repeated more than $1,000$ times and the end results are
the average results of the frequency analysis algorithm applied to
each of the real view candidates.

Moreover, evaluating the \emph{utility} preservation and studying the \emph{scalability} of using \emph{ORAM} in our scheme are respectively discussed in Appendix~\ref{app:over1} and ~\ref{app:over2}.

We conduct all experiments on a
machine running Windows with an Intel(R) Core(TM) i7-6700 3.40 GHz CPU, 4
GB Memory, and 500 GB storage.

\subsection{Results}
\subsubsection{Information Leakage Analysis}
First, the numerical results of the indistinguishability parameter
$\epsilon$ under different adversary's knowledges are depicted in
Figure~\ref{figepsasl4}(b). Those results correspond to three different
cases, i.e., when addresses are grouped based on (1) only the first
octet ($136$ groups), (2) the first and the second octets ($417$
groups), and (3) the first three octets ($506$ groups).  As we can see
from the results, $\epsilon$ decreases (meaning more privacy) as the
number of prefix groups increases, and it increases as the amount of
adversarial knowledge increases.

We next validate those numerical results through experiments in
Figure~\ref{fig:leakage23}. Specifically, we first analyze the
behavior of our second multi-view scheme (introduced in
Section~\ref{scheme2}) before comparing the two schemes in
Appendix~\ref{expappen}.  Figure~\ref{fig:leakage23} presents
different facets of information leakage  when our approach is
applied in various grouping cases. The results in
Figure~\ref{fig:leakage23} are for adversaries who has knowledge of no
more than most $50\%$ of the prefix groups (Figure~\ref{fig:leakage3}
in Appendix~\ref{app:over} presents the more extreme cases for the
same experiments, i.e., up to $100\%$ knowledge). The analysis of
these figures is detailed in the following.

\vspace{0.05in}

\noindent\textbf{Effect of the number of prefix groups:} As we discuss earlier,
three different IP grouping cases are
studied. Figures~\ref{fig:leakage23} (a) and (d) shows respectively
the results of packet leakage and number of real view candidates when
$d=136$. As the numerical results in
Figure~\ref{fig:detailed-approach1} anticipates, because the fraction
$d/D=0.154$ is relatively low, the indistinguishability of generated
views diminishes specially for stronger adversary
knowledges. Consequently, the adversary discards more views and the
rate of leakage increases, compared with Figures~\ref{fig:leakage23}
(b), (e) and Figures~\ref{fig:leakage23} (c), (f) for which the
fraction $d/D$ are $0.47$ and $0.57$, respectively. In particular, for
the worst case of $50\%$ adversary knowledge and when the number of
views is less than $50$, we can verify that the number of real view
candidates for case (1) remains $1$ resulting in packet leakage
comparable to that of CryptoPAn.

\vspace{0.05in}

\noindent\textbf{Effect of the number of
  views:} As it is illustrated in the figure, increasing the number of
views always improves both the number of real view candidates and the
packet leakages. All the figures for real view candidates evaluation,
show a near linear improvement where the slope of this improvement
inversely depends on the adversary's knowledge. For the packet
leakages, we can note that the improvement converges to a small packet
leakage rate under a large number of views. This is reasonable, as
each packet leakage result is an average of leakages in all the real
view candidates. However, since each of the fake views leaks a certain
amount of information, increasing the number of views beyond a certain
value will no longer affect the end result. In other words, the packet
leakage converges to the average of leakages in the (fake) real view
candidates. Finally, the results show that our proposed scheme can
more efficiently improve privacy by (1) increasing the fraction $d/D$
(\textit{number of views/number of distinct addresses}) or (2)
increasing the number of views. The first option may affect utility
(since inter-group prefix relations will be removed), while the second
option is more aligned with our objective of trading off privacy with
computation.

\subsubsection{Computational Overhead Evaluation}
\label{computcost}
We evaluate the computational overhead incurred by our
approach. Figure~\ref{fig:time-computation1} shows the time required
by our scheme in each grouping cases, when the number of views varies
for a trace including one million packets. We observe that, when the
number of views increases, the computational overhead increases near
linearly. However, each case shows a different slope depending on the
number of groups. This is reasonable as our second scheme generates
key vectors with a larger number of elements for more groups, which
leads to applying CryptoPAn for more iterations (see complexity
analysis in Appendix~\ref{complexityanalysis}). Finally, linking this
figure to the information leakage results shown in
Figure~\ref{fig:leakage23} demonstrates the trade-off between privacy
and computational overhead.

\section{Related Work}
\label{sec:related}
In the context of anonymization of network traces, as surveyed
in~\cite{Mivule}, many solutions have been
proposed~\cite{flaim,ref6,brekene1,paxon1,riboni}.  Generally, these
may be classified into different categories, such as
\textit{enumeration}~\cite{Farah}, \textit{partitioning}~\cite{Slagell}, and
\textit{prefix-preserving} ~\cite{ Xu,Gattani}. These methods include removing rows or attributes,
suppression, and generalization of rows or
attributes~\cite{challengeof}. Some of the solutions~\cite{Ribeiro,paxon1} are designed to address specific attacks and are generally based on the permutation of some fields in the network trace to blur the adversary's knowledge. Later studies
either prove theoretically~\cite{brekene2} or validate empirically~\cite{burkhart} that those works may be defeated by semantic attacks.

As our proposed anonymization solution
fall into the category of prefix-preserving solutions, which aims to
improve the utility, we review in more details some of the proposed
solutions in this category. First effort to find a prefix preserving anonymization was done by Greg Minshall~\cite{greg} who developed TCPdpriv which is a table-based approach that generates a function randomly. Fan et al.~\cite{PP} then developed CryptoPAn with a completely cryptographic approach. Several publications~\cite{brekene1}, ~\cite{Ribeiro,paxon1} have then raised the vulnerability of this scheme against semantic attacks which motivated query based~\cite{mcsherry} and bucketization based~\cite{riboni} solutions. In the following we review those works in more details.

Among the works that address such semantic attacks, Riboni et
al.~\cite{riboni} propose a (k,j)-obfuscation methodology applied to
network traces. In this method, a flow is considered obfuscated if it
cannot be linked, with greater assurance, to its (source and
destination) IPs. First, network flows are divided into either
confidential IP attributes or other fields that can be used to attack.
Then, groups of $k$ flows having similar fingerprints are first
created, then bucketed, based on their fingerprints into groups of
size $j<k$.  However, utility remains a challenge in this solution, as
the network flows are heavily sanitized, i.e., each flows is blurred
inside a bucket of $k$ flows having similar fingerprints. An
alternative to the aforementioned solutions, called \textit{mediated
  trace analysis}~\cite{mediated1,mediated2}, consists in performing
the data analysis on the data-owner side and outsourcing analysis
reports to researchers requesting the analysis. In this case, data can
only be analyzed where it is originally stored, which may not always
be practical, and the outsourced report still needs to be sanitized
prior to its outsourcing~\cite{mcsherry}. In contrast to those
existing solutions, our approach improves both the privacy and utility
at the cost of a higher computational overhead. Table~\ref{relatedss}, summarizes the most important network trace anonymization schemes, over past twenty
years\cite{Mivule} and their main characteristics.
\begin{table}[!t]
	\caption{Summary of proposed network traces anonymization in literature}
	\label{relatedss}
	\centering
	\begin{adjustbox}{max width=3.3in}
    \LARGE
		\begin{tabular}{c|c|c }
			\hline
			Authors & Privacy against semantic attacks & Utility \\
			\hline
			Slagell et al. \cite{slagell} & Violated & Prefix preserving\\
			\hline
McSherry et al. \cite{mcsherry} & Preserved& Noisy aggregated results\\
			\hline
Pang et al. \cite{paxon1} & Violate & Partial prefix preserving\\
\hline
Riboni et al. \cite{riboni} &Preserved & Heavily sanitized\\
\hline
Ribeiro et al. \cite{Ribeiro} &Violated & Partial prefix preserving\\
\hline
Mogul et al. \cite{mediated1} &Violated & Aggregated results\\
\hline
\end{tabular}
	\end{adjustbox}
\end{table}

The last step of our solution requires data owner to privately
retrieve an audit report of the real view, which can be based on
existing private information retrieval (PIR) techniques.  A PIR
approach usually aims conceal the objective of all queries independent
of all previous queries~\cite{orpir,pir2}. Since the sequence of
accesses is not hidden by PIR while each individual access is hidden,
the amortized cost is equal to the worst-case cost~\cite{orpir}. Since
the server computes over the entire database for each individual
query, it often results in impracticality for large databases. On the
other hand, ORAM~\cite{oram1} has verifiably low amortized
communication complexity and does not require much computation on the
server but rather periodically requires the client to download and
reshuffle the data~\cite{orpir}.  For our multi-view scheme, we choose
ORAM as it is relatively more efficient and secure, and also the
client (data owner in our case) has sufficient computational power and
storage needed to locally store a small number of blocks (audit
reports in our case) in a local stash.
\section{Conclusion}
\label{sec:conclusion}
In this paper, we have proposed a multi-view anonymization approach
mitigating the semantic attacks on CryptoPAn while preserving the
utility of the trace. This novel approach shifted the trade-off from
between privacy and utility to between privacy and computational cost;
the later has seen significant decrease with the advance of technology,
making our approach a more preferable solution for applications that
demand both privacy and utility. Our experimental results showed that
our proposed approach significantly reduced the information leakage
compared to CryptoPAn. For example, for the extreme case of adversary
pre-knowledge of 100\%, the information leakage of CryptoPAN was 100\%
while under approach it was still less than 10\%. Besides addressing
various limitations discussed in Appendix~\ref{ffaa}, our future works will
adapt the idea to improve existing privacy-preserving solutions in
other areas, e.g., we will extend our work to the multi-party problem
where several data owners are willing to share their traces to
mitigate coordinated network reconnaissance by means of distributed
(or inter-domain) audit~\cite{sepia}.
\section{ACKNOWLEDGMENTS}
The authors thank the anonymous reviewers, for their valuable comments. We appreciate
Momen Oqaily’s support in the implementation. This work is partially
supported by the Natural Sciences and Engineering Research
Council of Canada and Ericsson Canada under CRD Grant
N01823. The research of Yuan Hong is partially supported by the National Science Foundation
under Grant No. CNS-1745894 and the WISER ISFG grant.


\appendix
\section{Proofs}
\subsection{Proof of Theorem~\ref{thm:reverese-prefic-preserving}}
\label{subsec:proof1}
\begin{proof}\ref{thm:reverese-prefic-preserving}
We must show that $c=a$. To do so, we use induction:
\begin{align*}
\begin{split}
&c_{1}  = b_{1} \oplus f_{0}(\lambda) , \  \  b_{1} =  a_{1} \oplus f_{0}(\lambda) \\
&\Rightarrow c_{1} = a_{1} \oplus f_{0}(\lambda)\oplus f_{0}(\lambda) = a_{1}
\end{split}
\end{align*}
where $\lambda$ is empty string and $f_{0}(\lambda)$ is a constant bit in $\{0,1\}$ that depends only on padding function and cryptographic key $K$. Assume $\forall j < k;  \ \ c_{j} = a_{j}$, thus:
\begin{align*}
&c_{k} = b_{k} \oplus f_{k-1}(c_1 \cdots c_{k-1}) \\
& \text{and:} \  \forall j < k;  \ \ \ \ c_{j} = a_{j}, \   b_{k} =  a_{k} \oplus f_{k-1}(a_1 \cdots a_{k-1}) \\
&\Rightarrow c_{k} = a_{k} \oplus f_{k-1}(a_1 \cdots a_{k-1}) \oplus f_{k-1}(a_1 \cdots a_{k-1})= a_{k}
\end{align*}
\end{proof}
\subsection{Proof of Lemma~\ref{lem:jj}}
\label{prooflem1}
\begin{proof}
Suppose there exists algorithm $\mathsf{A}$ which returns the smallest set of key vectors $V_{s}$ to reverse the seed trace and obtain the minimum number of real view candidates, given our setup. Also denote by $V$ the set of those key vectors if the adversary follows scheme~\RN{2}. We now show that if $(2d-2)^D>N$ holds then we have $V=V_{s}$. First, note that the key indices of different distinct addresses in $\mathcal{L}^*_0$ is $PRNG(d,D,0)$. Therefore, the adversary has to guess $V=\mathcal{C}^*-PRNG(d,D,0)$. However, note that elements of $V$ are in $[-d+1,d-1]$ and there will be $(2d-2)^D$ different combinations for $V$. Thus to minimize this number, the adversary has to use the outsourced parameters which means we have $\mathsf{A}(\mathcal{L}^*_0,K,V_1,\cdots,V_N,\mathcal{S}_{\alpha})=V_{s}$. However, we showed earlier that all these inputs are trivial leakage. Therefore, if $(2d-2)^D>N$  holds, we have $V= V_{s}$.
\end{proof}
\section{Experiments}
\label{expappen}
In this section, using experiments, we measure the security of the proposed approach against very strong adversaries. In addition, we evaluate the utility of the approach using two real network analyses. Finally, we justify the choice of ORAM in our setup using a comprehensive study on the scalability of ORAM in the literature.
\subsection{Privacy Evaluation against Very Strong Adversaries}
Figure~\ref{fig:leakage3} shows the leakage and the real view candidates results for stronger adversaries ($\alpha \in [60,100]$). Note that in this figures, we only show results for case (2) and (3) as results in case (1) does not show a significant improvement compared with CryptoPAn results because the multi-view approach with fraction of $d/D=0.154$ cannot defeat the adversary's knowledge ($\epsilon>16$).
\label{app:over}
\subsection{Utility Evaluation Using Real-life Network Analytics}
\label{app:over1}
Figure~\ref{fig:leakage4} shows the results of two different network analytics over the original trace (1M records), the real view  and one of the fake views generated in our multi-view solution. In the first experiment, we present IP distribution~\cite{IPdist} in the trace; reporting the number of distinct addresses within each subnet (IP group). We compare the distribution of distinct IP addresses inside the aforementioned three traces for both ~\emph{temporal} distribution; if subnets are indexed based on their time stamps; and ~\emph{cardinality-based} distribution result; if subnets are indexed based on their cardinalities. We found that our results (both distributions) generated from the original trace and the real view are identical (see Figure~\ref{fig:leakage4}(a)). This is reasonable because the real view is a prefix preserving mapping of IPs that keeps the fp-QI attributes intact (preserving both distributions). Moreover, the cardinality based distribution result generated from the fake view is identical to those in the original trace and the real view (see Figure~\ref{fig:leakage4}(c)). Note that the later is resulted from the indistinguishability of our multi-view solution.

In the second experiment, we present a packet-level analytic~\cite{mcsherry}. In particular, Figure~\ref{fig:leakage4}(d,e) shows the~\emph{empirical cumulative distribution function} results for the three traces. Our results clearly show that the original trace and our scheme results are identical as multi-view will not have any impact on fingerprinting quasi identifier attributes.
\begin{figure*}[ht]
	\begin{center}
		\includegraphics[width=0.6\textwidth, viewport= 5 92 900 640, clip]{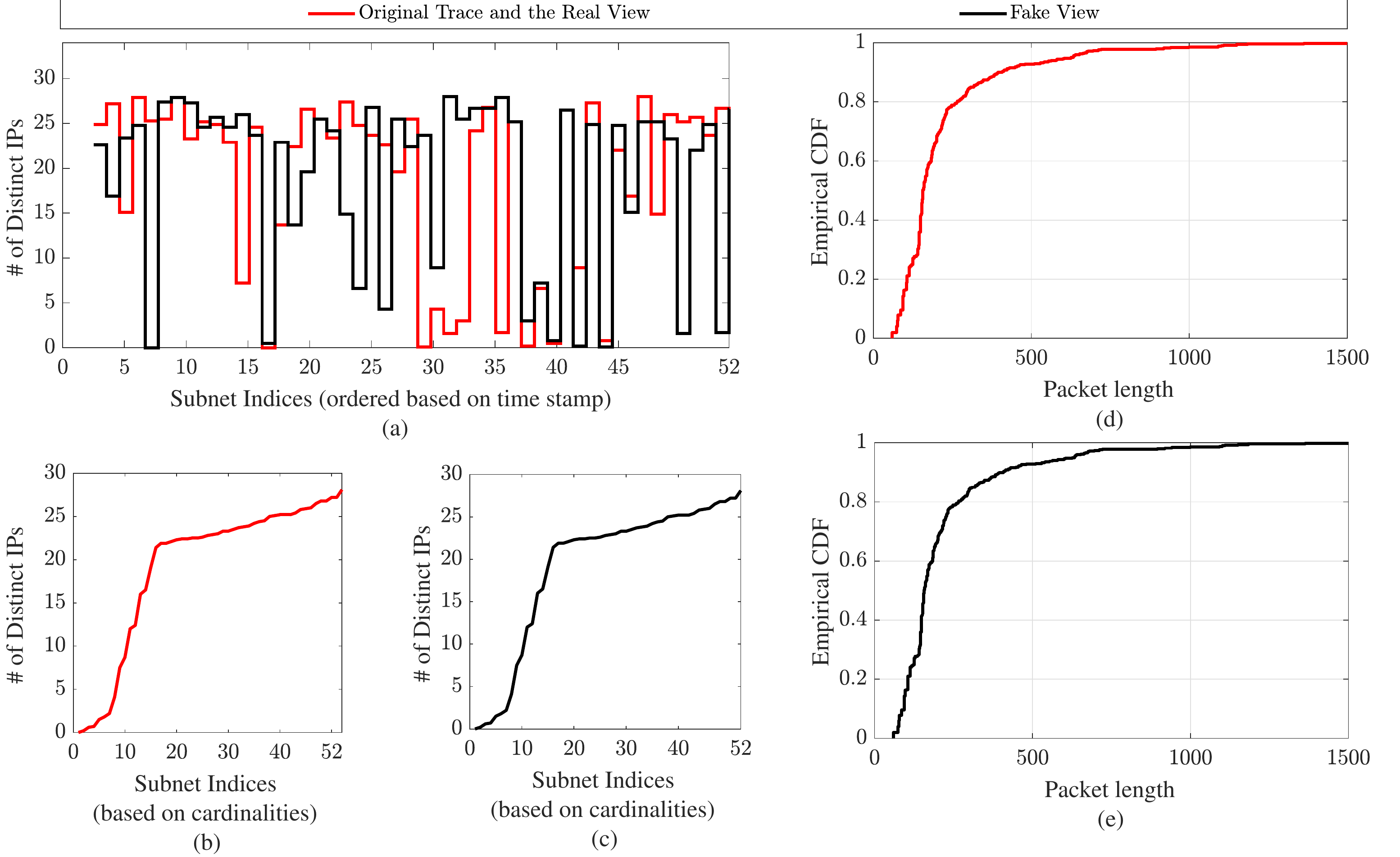}
		\caption{Distribution of distinct IP addresses in different subnets (IP groups) (out of 1M) (a) for the original trace, the real view and one of the fake views based on the order they appear in the trace (temporal distribution), (b) for the original trace and the real view and (c) for the fake view, based on the cardinalities of the subnets in an ascending order (cardinality-based distribution). Empirical CDF for the packet lengths in (e),(f) the original trace and the real view, and the fake view, respectively.}
		\label{fig:leakage4}
	\end{center}
\end{figure*}
\begin{figure}[ht]
	\begin{center}
		\includegraphics[width=0.83\textwidth, viewport= 65 5 1360 614, clip]{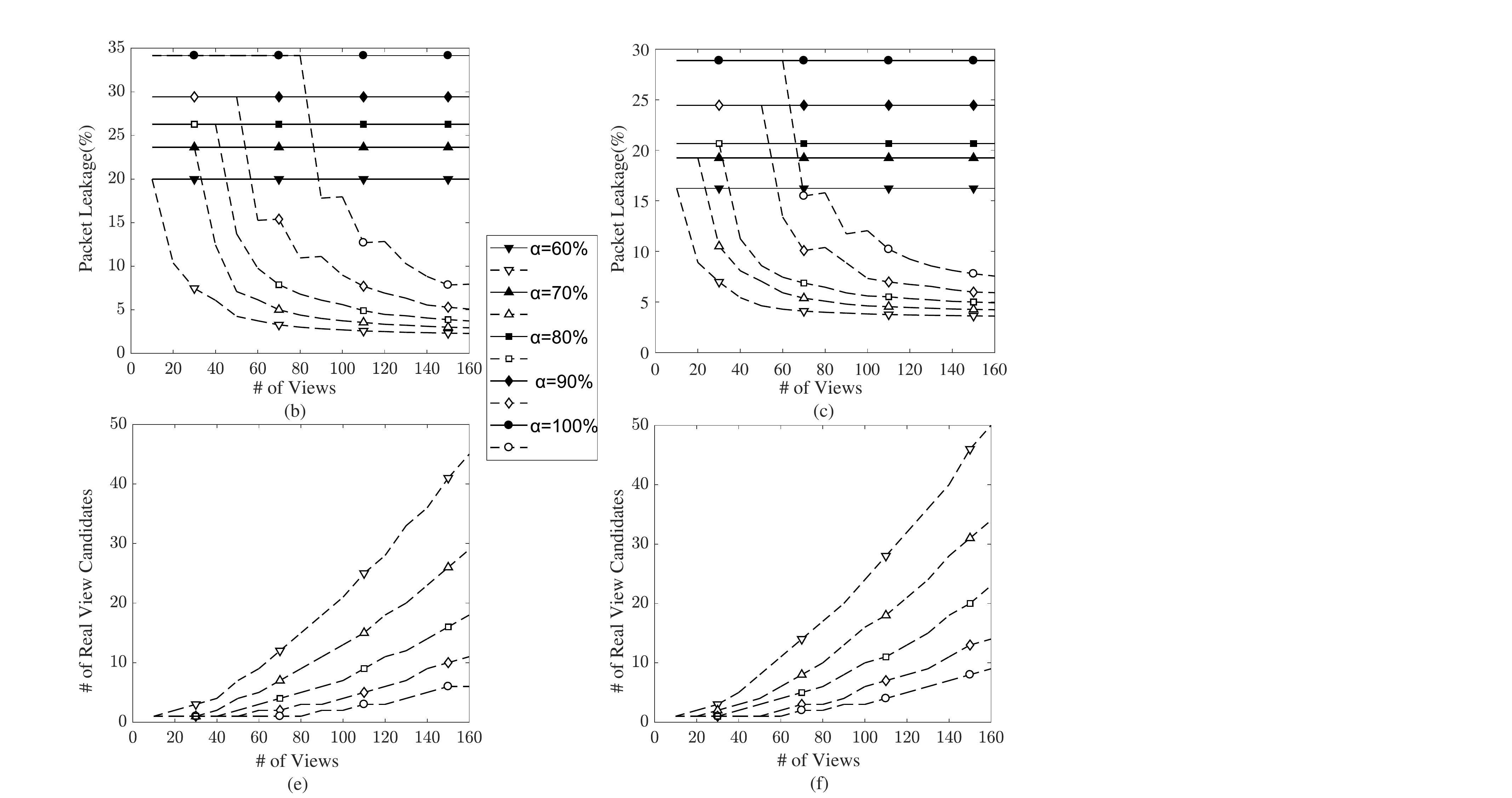}
		\caption{Percentage of the compromised packets (out of 1M) and number of real view candidates when number of views and the adversary knowledge vary and for  case (1) Figures (b),(e) (2) Figures (c),(f) where legends marked by \emph{CP} denote the CryptoPAn result whereas those marked by \emph{MV} denote the multi-view results}
		\label{fig:leakage3}
	\end{center}
\end{figure}
\begin{figure}[ht]
	\begin{center}
		\includegraphics[width=0.97\textwidth,  viewport= -90 350 1360 650, clip]{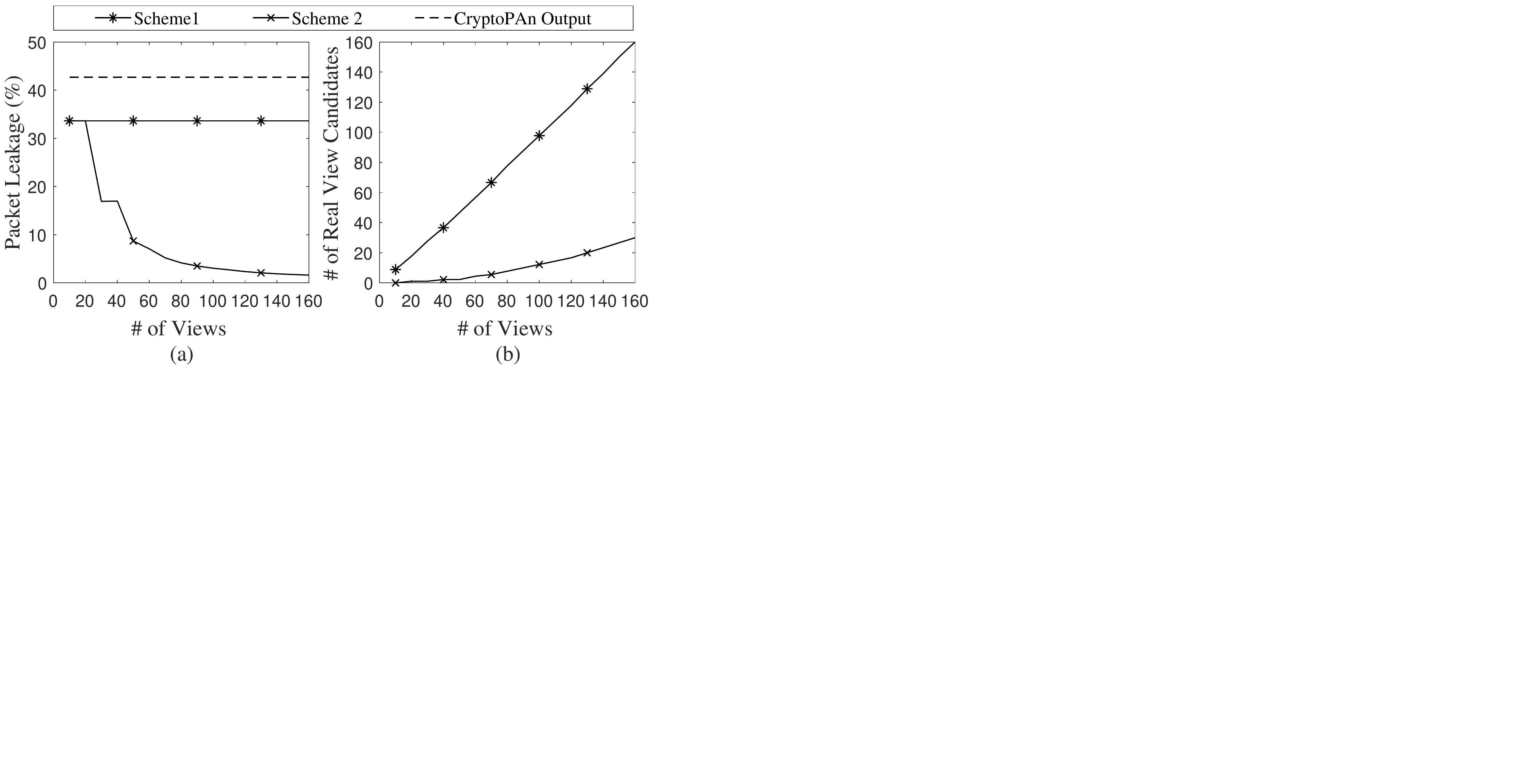}
		\caption{Comparison between scheme~\RN{1} and scheme~\RN{2} with $137$ partitions (prefix groups based on first octet sharing). Figure (a): Percentage of the compromised packets (out of 1M) and Figure (b): Number of real view candidates for $40\%$ adversary knowledge}
		\label{fig:a233}
	\end{center}
\end{figure}
\subsection{Multi-view and the Scalability of ORAM}
\label{app:over2}
In practice, we expect analysis reports would have significantly smaller sizes in comparison to the views, and considering the one round communication with ORAM $(O(log N)$-complexity), we believe the solution would have acceptable scalability. Experiments using our dataset and existing ORAM implementation (an implementation~\cite{dong} of non-recursive Path-ORAM~\cite{path} has been made public) would further confirm this. We generated various set of analyses reports using \emph{snort}~\cite{snort}, and we found that for our dataset the size of audit reports are in the range of KB which is perfect to be used in fast ORAM protocols, e.g., Path-ORAM. Specifically, for Path-ORAM, Figure 5 (b) in~\cite{dong} shows a less than 1MB communication overhead for the worst-case cost of up to $2^{24}$ number of blocks of size 4KB.
 \section{Algorithms}
  \begin{algorithm}[ht]
	\SetAlgoLined
	\caption{Data owner: Trace anomymization (scheme~\RN{1})}
	\label{alg:data-owner-actions}
	\textbf{Input:}\\
	\ \ \ \ \ \ \   $\mathcal{L}$: Original network trace\\
	\ \ \ \ \ \ \	$K_0$, $K$:  Cryptographic keys \\
    \ \ \ \ \ \ \	$d$: Number of  prefix groups\\
	\ \ \ \ \ \ \   $partition(\mathcal{L},d)$: IP partitioning\\
	\ \ \ \ \ \ \	$r$: Iteration number of the real view.\\
	\ \ \ \ \ \ \	$V$: Random vectors, of size $d$ \\
	\textbf{Output:} \\
	\ \ \ \ \ \ \	$\mathcal{L}^*$: Anonymized trace to be outsourced \\
	\textbf{Function: anonymize} ($\mathcal{L}, d, K_0, K, r, V$) \\
	\textbf{begin} \\
	1 \ \ \ \ \ \ \	$\mathcal{L}:=PP(\mathcal{L},K_0)$ \\
	2  \ \ \ \ \ \ \	$V_{0}:= -r \cdot V$ \\
	3 \ \ \ \ \ \ \	$P:=partition(\mathcal{L},d)$ \\
	4 \ \ \ \ \ \ \	$\mathcal{L}^*:=\phi$\\
	5 \ \ \ \ \ \ \	\textbf{foreach} \ \ $P_{i} \in P$ \ \ \ \textbf{do:} \\
	6 \ \ \ \ \ \ \	\ \ \	$\mathcal{L}_{i} := $ GetFlows{($\mathcal{L},P_{i}$)} \\
	7 \ \ \ \ \ \ \ \ \ \	$\mathcal{L}^*_{i} := PP(\mathcal{L}_{i},V_0(i),K)$\\
	8 \ \ \ \ \ \ \ \ \ \	$\mathcal{L}^* := \mathcal{L}^* \cup \mathcal{L}^*_{i}$\\

	9 \ \ \ \ \ \ \ \ \textbf{end} \\
    10 \ \ \ \  	\textbf{return} $\mathcal{L}^*$, $K, V, N$ \\
	\textbf{end}\\
\end{algorithm}
\begin{algorithm}[ht]
	\SetAlgoLined
	\caption{Analyst: Network trace analysis (scheme~\RN{1})}	
	\label{alg:analyst-actions}
	\textbf{Input:}\\
	\ \ \ \ \ \ \	$\mathcal{L}^*$: Seed trace\\
	\ \ \ \ \ \ \	$N$: Number of iterations requested by\\ 	\ \ \ \ \ \ \ 	\ \ \ \ \  data owner\\
    \ \ \ \ \ \ \	$d$: Number of  prefix groups\\
	\ \ \ \ \ \ \	$partition(\mathcal{L^*},d)$: IP partitioning\\
	\ \ \ \ \ \ \	$K$: Outsourced key \\
	\ \ \ \ \ \ \	$V$: Vector of size $d$ defined by data owner\\
	\ \ \ \ \ \ \	$CV(\mathcal{L}^*$): Compliance verification \\
	\textbf{Output:}\\
	\ \ \ \ \ \ \	$\Gamma_{i}$: Analysis report of $i^{th}$ view $\mathcal{L}^*_{i}$,\\ 	\ \ \ \ \ \ \ \ \ \ \ \  $i\in\{1,2,\ldots,,N\}$ \\
	\textbf{Function: analysis} ($\mathcal{L}^*, d, partition(\mathcal{L}^*,d), K, N, V$)\\
	\textbf{begin} \\
	1\ \ \ \ \ \ \	$P:=partition(\mathcal{L^*},d)$ \\
	2\ \ \ \ \ \ \	\textbf{for} $i=1:N$ \ \ \textbf{do:} \\
	3\ \ \ \ \ \ \ \ \ \	$\mathcal{L}_{i}^*:=\phi$\\
	4\ \ \ \ \ \ \ \ \ \	\textbf{foreach} \ \ $P_{j} \in P$ \ \ \ \textbf{do:}\\
	6\ \ \ \ \ \ \ \ \ \ \ \ \	$\mathcal{L}_{i,j}:=$ GetFlows{($\mathcal{L^*}_{i-1},P_{j}$)} \\
	
	7\ \ \ \ \ \ \ \ \ \ \ \ \	$\mathcal{L}^*_{i,j}:= PP(\mathcal{L}_{i,j},V(j), K)$\\
	8\ \ \ \ \ \ \	\ \ \ \ \ \	$\mathcal{L}^*_{i}:= \mathcal{L}^*_{i} \cup \mathcal{L}^*_{i,j}$\\
	9\ \ \ \ \ \ \	\ \ \	\textbf{end}\\
	10\ \ \ \ \ \ \	\ \ 	$\Gamma_{i}:=CV(\mathcal{L}^*_{i}$) \\
	11\ \ \ \ \ \ \	\ \ 	\textbf{return} $\Gamma_{i}$\\
	12\ \ \ \ \ \ \	\textbf{end}\\
	\textbf{end}
\end{algorithm}

\label{algs}
Following Algorithms are summarized versions of the data owner's and the analyst's roles in our multi-view scheme presented in section~\ref{sec:instance}. \\
\textbf{Algorithm~\ref{alg:data-owner-actions}:} The data owner's actions (scheme~\RN{1}).\\
\textbf{Algorithm~\ref{alg:analyst-actions}:} The analyst's actions (scheme~\RN{1}).\\
\textbf{Algorithm~\ref{alg:data-owner-actions1}:} The data owner's actions (scheme~\RN{2}).\\
\textbf{Algorithm~\ref{alg:analyst-actions1}:} The analyst's actions (scheme~\RN{2}).
\section{Complexity Analysis}
\label{complexityanalysis}
Here, we discuss the overhead analysis, from both the data owner's and the data analyst's side. In particular, table~\ref{relatedss1} summarizes the overhead for all the action items in the data owner side. Here, $C(n)$ is the computation overhead of CryptoPAn and $D$ is the number of the distinct IP addresses. Finally, table~\ref{relatedss2} summarizes the overhead for all the action items in the data analyst side where $N.CV(n)$ is the cost of $N$ times verifying the compliances (auditing).
\begin{table}[ht]
	\caption{Overhead on the data owner side}
	\label{relatedss1}
	\centering
	\begin{adjustbox}{max width=3.2in}
    \LARGE
		\begin{tabular}{c|c|c }
			\hline
			Blocks in Multi-view & Computation Overhead & Communication Overhead \\
			\hline
			 Initial anonymization& $C(n)$ & ---\\
			\hline
 Migration function & $O(n log^{n})+\frac{\sum^d_{i=1} \mathcal{C}(i)}{d} C(n)$& --- \\
			\hline
Prefix grouping& --- & ---\\
\hline
Index generator& $ N.O(D)$ & $ N.O(D)$\\
\hline
Seed trace & $\frac{\sum^D_{i=1} V_0(i)}{D} C(n)$ & $O(n)$ \\
\hline
Report retrieval (ORAM) &---& $O(log^N)  \omega (1)$\\
\hline
		\end{tabular}
	\end{adjustbox}
\end{table}
\begin{table}[ht]
	\caption{Overhead on the data analyst side}
	\label{relatedss2}
	\centering
	\begin{adjustbox}{max width=3.2in}
    \LARGE
		\begin{tabular}{c|c|c }
			\hline
			Blocks in Multi-view & Computation Overhead & Communication Overhead \\
			\hline
 Seed view & ---& $O(n)$ \\
			\hline
N views generation& $\frac{\sum^N_{i=1} \sum^D_{j=1} V_{i}(j)}{D} C(n)$ & ---\\
\hline
Compliance verification (Analysis) & $N.CV(n)$& --- \\
\hline
		\end{tabular}
	\end{adjustbox}
\end{table}
\begin{algorithm}[ht]
	\SetAlgoLined
	\caption{Data owner: Trace anomymization (scheme~\RN{2})}
	\label{alg:data-owner-actions1}
	\textbf{Input:}\\
	\ \ \ \ \ \ \   $\mathcal{L}$: Original network trace\\
	\ \ \ \ \ \ \	$K_0$, $K$:  Cryptographic keys \\
	\ \ \ \ \ \ \	$D$: Number of  IPs\\
    \ \ \ \ \ \ \	$d$: Number of  prefix groups\\
    \ \ \ \ \ \ \   $Migration(\mathcal{L},d)$: Migration function \\
	\ \ \ \ \ \ \   $partition(\mathcal{L},D)$: IP partitioning\\
	\ \ \ \ \ \ \	$r$: Iteration number of the real view.\\
	\ \ \ \ \ \ \	$V_0,V_1,\cdots,V_N$: Vectors of size $D$ defined by data owner\\
	\textbf{Output:} \\
	\ \ \ \ \ \ \	$\mathcal{L}^*$: Anonymized trace to be outsourced \\
	\textbf{Function: anonymize} ($\mathcal{L}, d, D, K_0, K, r, V$) \\
	\textbf{begin} \\
	1-1 \ \ \ \ 	$\mathcal{L}:=PP(\mathcal{L},K_0)$ \\
    1-2 \ \ \ \ 	$\mathcal{L}:=Migration(\mathcal{L},d)$ \\
	2  \ \ \ \ \ \ \	$V_{0}:= -r \cdot V$ \\
	3 \ \ \ \ \ \ \	$P:=partition(\mathcal{L},D)$ \\
	4 \ \ \ \ \ \ \	$\mathcal{L}^*:=\phi$\\
	5 \ \ \ \ \ \ \	\textbf{foreach} \ \ $P_{i} \in P$ \ \ \ \textbf{do:} \\
	6 \ \ \ \ \ \ \	\ \ \	$\mathcal{L}_{i} := $ GetFlows{($\mathcal{L},P_{i}$)} \\
	7 \ \ \ \ \ \ \ \ \ \	$\mathcal{L}^*_{i} := PP(\mathcal{L}_{i},V_0(i),K)$\\
	8 \ \ \ \ \ \ \ \ \ \	$\mathcal{L}^* := \mathcal{L}^* \cup \mathcal{L}^*_{i}$\\

	9 \ \ \ \ \ \ \textbf{end} \\
    10 \ \ \ \  	\textbf{return} $\mathcal{L}^*$, $K, V_1,V_2,\cdots,V_N, N$ \\
	\textbf{end}\\
\end{algorithm}
\begin{algorithm}[ht]
	\SetAlgoLined
	\caption{Analyst: Network trace analysis (scheme~\RN{2})}	
	\label{alg:analyst-actions1}
	\textbf{Input:}\\
	\ \ \ \ \ \ \	$\mathcal{L}^*$: Seed trace\\
	\ \ \ \ \ \ \	$N$: Number of iterations requested by\\ 	\ \ \ \ \ \ \ 	\ \ \ \ \  data owner\\
        \ \ \ \ \ \ \	$d$: Number of  prefix groups\\
	\ \ \ \ \ \ \	$D$:  Number of partitions \\
	\ \ \ \ \ \ \	$partition(\mathcal{L^*},D)$: IP partitioning\\
	\ \ \ \ \ \ \	$K$: Outsourced key \\
	\ \ \ \ \ \ \	$V_0,V_1,\cdots,V_N$: Vectors of size $D$ defined by data owner\\
	\ \ \ \ \ \ \	$CV(\mathcal{L}^*$): Compliance verification \\
	\textbf{Output:}\\
	\ \ \ \ \ \ \	$\Gamma_{i}$: Analysis report of $i^{th}$ view $\mathcal{L}^*_{i}$,\\ 	\ \ \ \ \ \ \ \ \ \ \ \  $i\in\{1,2,\ldots,,N\}$ \\
	\textbf{Function: analysis} ($\mathcal{L}^*, D, partition(\mathcal{L}^*,D), K, N, V_1,V_2,\cdots,V_N$)\\
	\textbf{begin} \\
	1\ \ \ \ \ \ \	$P:=partition(\mathcal{L^*},D)$ \\
	2\ \ \ \ \ \ \	\textbf{for} $i=1:N$ \ \ \textbf{do:} \\
	3\ \ \ \ \ \ \ \ \ \	$\mathcal{L}_{i}^*:=\phi$\\
	4\ \ \ \ \ \ \ \ \ \	\textbf{foreach} \ \ $P_{j} \in P$ \ \ \ \textbf{do:}\\
	6\ \ \ \ \ \ \ \ \ \ \ \ \	$\mathcal{L}_{i,j}:=$ GetFlows{($\mathcal{L^*}_{i-1},P_{j}$)} \\
	
	7\ \ \ \ \ \ \ \ \ \ \ \ \	$\mathcal{L}^*_{i,j}:= PP(\mathcal{L}_{i,j},V_{i}(j), K)$\\
	8\ \ \ \ \ \ \	\ \ \ \ \ \	$\mathcal{L}^*_{i}:= \mathcal{L}^*_{i} \cup \mathcal{L}^*_{i,j}$\\
	9\ \ \ \ \ \ \	\ \ \	\textbf{end}\\
	10\ \ \ \ \ \ \	\ \ 	$\Gamma_{i}:=CV(\mathcal{L}^*_{i}$) \\
	11\ \ \ \ \ \ \	\ \ 	\textbf{return} $\Gamma_{i}$\\
	12\ \ \ \ \ \ \	\textbf{end}\\
	\textbf{end}
\end{algorithm}

\end{document}